\documentclass[journal,twocolumn]{IEEEtran} 
\usepackage{graphicx}
\usepackage{bm}
\usepackage{url}
\usepackage{mystyle} 
\usepackage{cite}
\usepackage{amsfonts,amssymb,amsmath}
\allowdisplaybreaks
\centerfigcaptionstrue
\def\defeq{\stackrel{\triangle}{=}}
\def\bI{{\mathbf I}}                % Identity matrix   
\def\bH{{\mathbf H}}                % Channel matrix
\def\bHeq{{\mathbf{H_{e}}}}     % equivalent chs
\def\bx{{\mathbf x}}                % transmitted vector
\def\by{{\mathbf y}}                % received vector
\def\bn{{\mathbf n}}                % Noise
\def\bW{{\mathbf W}}                % Equalizer notation
\def\bfunc{{\boldsymbol{f}}}        % f: account for both channel H and equalizer W
\def\prob{{\mathbb P}}              % Probability symbol
\def\expect{{\mathbb E}}            % Expectation symbol
\def\balpha{{\boldsymbol{\alpha}}}  % Alpha symbol
\def\bQ{{\boldsymbol{Q}}}  % Q symbol
\def\b0{{\boldsymbol{0}}}

\title{Diversity of MMSE MIMO Receivers}

\author{Ahmed Hesham Mehana, {\em Student Member, IEEE}, and Aria Nosratinia, {\em Fellow, IEEE}
\thanks{ The authors are with the Department
  of Electrical Engineering, the University of Texas at Dallas,
  Richardson, TX 75083-0688 USA, E-mail:
  ahmed.mehana@student.utdallas.edu; aria@utdallas.edu. }
}

\begin{document}

\maketitle

%-------------------------------- Abstract ---------------------------------------
\begin{abstract}

In most MIMO systems, the family of waterfall error curves, calculated
at different spectral efficiencies, are asymptotically parallel at
high SNR. In other words, most MIMO systems exhibit a single diversity
value for all fixed rates. The MIMO MMSE receiver does not follow this
pattern and exhibits a varying diversity in its family of error
curves. 
This work analyzes this interesting behavior of the MMSE MIMO receiver
and produces the MMSE MIMO diversity at all rates. The diversity of
the quasi-static flat-fading MIMO channel consisting of any arbitrary
number of transmit and receive antennas is fully characterized,
showing that full spatial diversity is possible if and only if the
rate is within a certain bound which is a function of the number of
antennas. For other rates, the available diversity is fully
characterized. At sufficiently low rates, the MMSE receiver has a
diversity similar to the maximum likelihood receiver (maximal
diversity), while at high rates it performs similarly to the
zero-forcing receiver (minimal diversity). Linear receivers are also
studied in the context of the MIMO multiple access channel
(MAC). Then, the quasi-static frequency selective MIMO channel is
analyzed under zero-padding (ZP) and cyclic-prefix (CP) block
transmissions and MMSE reception, and lower and upper bounds on
diversity are derived. For the special case of SIMO under CP, it is
shown that the above-mentioned bounds are tight.

\end{abstract}


%--------------------------------- Introduction ----------------------------------

\section{Introduction}

Linear receivers are widely used for their low complexity compared to
maximum likelihood (ML) receivers.  In the context of MIMO systems,
linear receivers such as the minimum mean square error (MMSE) receiver
are adopted in some of the emerging standards, e.g. IEEE 802.11n and
802.16e. Therefore the analysis of MMSE receivers is strongly
motivated by both theoretical and practical considerations.

A significant amount of research has focused on linear receivers,
however, their performance is not fully understood in the MIMO
channel. For instance, the distribution of the output
signal-to-interference-plus-noise ratio (SINR) of the linear MIMO
receiver is still unknown except in asymptotic regimes (large number
of antennas, and high/low
SNR)~\cite{LiCioffi:JIT06,Moustakas:VTC09,Boche:Eura07,Jiang:JIT07}.
The outage and diversity of MMSE receiver have also been a subject of
interest. It has been
observed~\cite{Onggosanusi:ICAS02,Hedayat:JSP07,Kumar:JIT09} that
while the MMSE receiver can extract the full spatial diversity of the
MIMO quasi-static channel at low rates, it does not enjoy this feature
at high rates.

Figure~\ref{Fig5} shows the outage probabilities (for various spectral
efficiencies $R$ bps/Hz) of MMSE and ML receivers
respectively. Clearly, one of the main differences between the two
characteristics is the slope of the error curves, i.e., the {\em
  diversity}. Figure~\ref{Fig5} shows that in a $2\times 2$ MIMO
system the ML receiver achieves diversity 4 at all rates. However, the
MMSE receiver diversity varies with the operating spectral
efficiency. From a system design perspective, obtaining the MMSE
diversity is important in order to understand the broad tradeoffs
involved in the determination of the operating point of the system and
predicting its performance.

\begin{figure*}
\begin{minipage}{3.5in}
\centering
\includegraphics[width=3.5in]{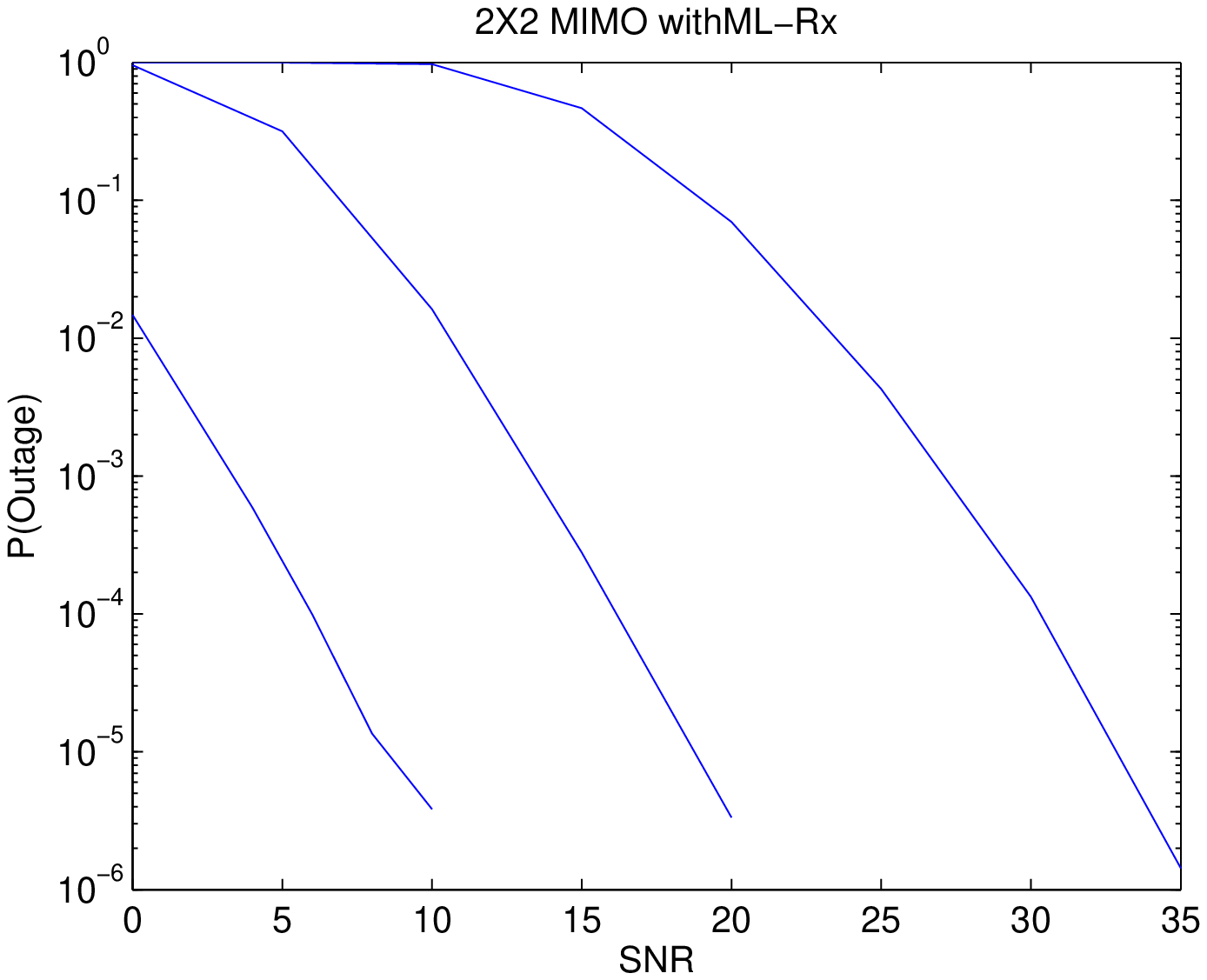}
\end{minipage}
\hfill
\begin{minipage}{3.5in}
\centering
\includegraphics[width=3.5in]{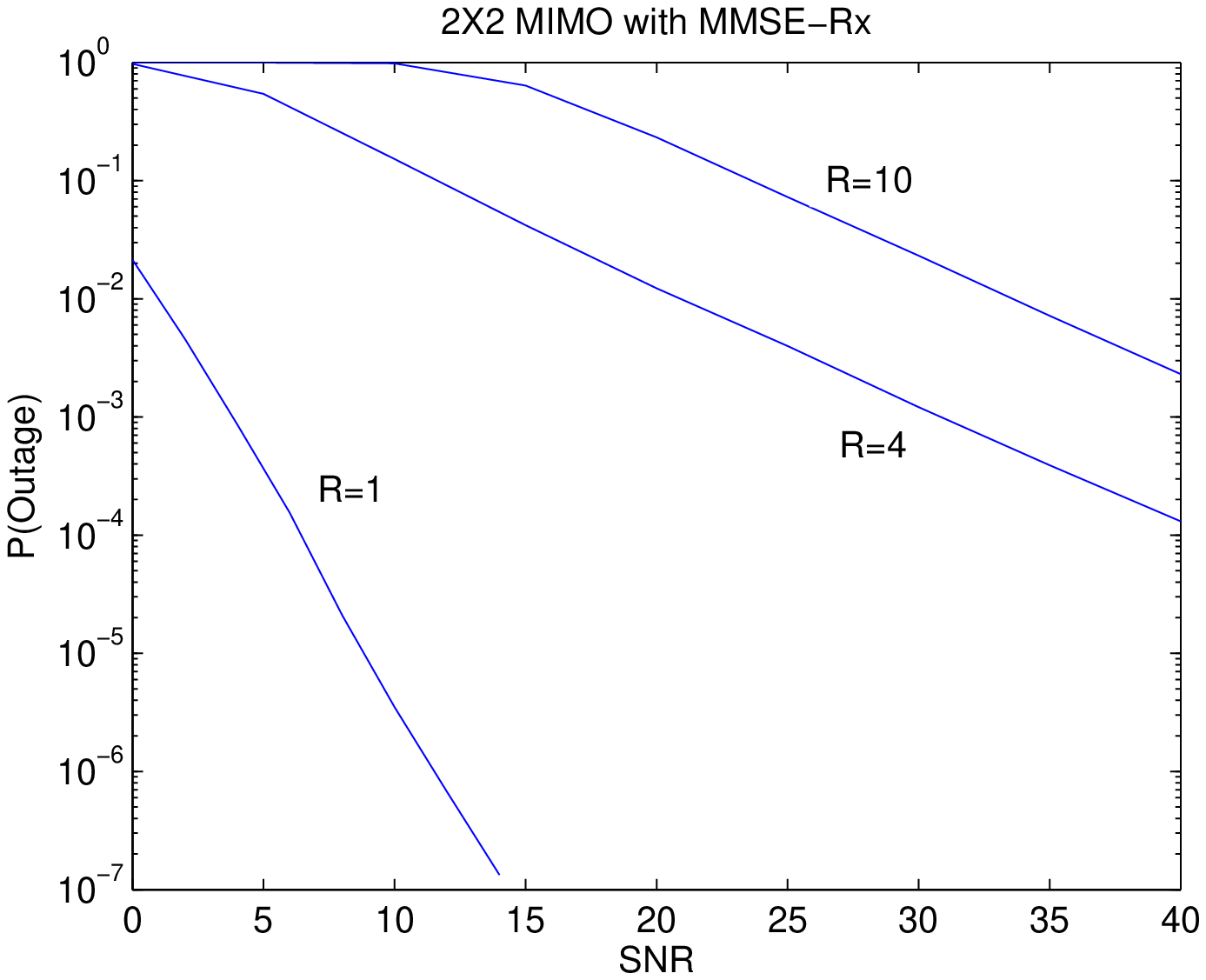}
\end{minipage}
\caption{Outage probability of ML receiver (left) and MMSE (right) with  $M=N=2$ antennas and for rates R= 1, 4, and 10 bps/Hz}
\label{Fig5}
\end{figure*}

In this work we seek answers for the following questions: when can the
MMSE receiver exploit the full diversity in MIMO channel? More
generally, how does the diversity of the MMSE receiver vary with the
system parameters such as spectral efficiency $R$, the number of
antennas, and in case of inter-symbol interference channel (ISI), the
channel memory?

The well-known and powerful framework of diversity-multiplexing
tradeoff (DMT) is not sufficient to answer the above questions,
because the DMT framework cannot distinguish between different
spectral efficiencies that correspond to the same multiplexing
gain. In the MIMO MMSE receiver, rates that correspond to the same
multiplexing gain can produce different diversities.
%% As mentioned earlier, in the MIMO MMSE receiver various spectral
%% efficiencies $R$ that correspond to the same multiplexing gain give
%% rise to different diversities.

We approach the problem of MMSE reception in MIMO flat fading channels
through a rate-dependent approximation of the outage probability and
then proceed with bounding the pairwise error probability (PEP) from
both sides using the outage. This leads to a closed-form expression
for the diversity-rate tradeoff which reveals the relationship between
diversity, spectral efficiency, and number of transmit and receive
antennas. The approximation of outage and PEP as functions of rate
requires more delicate handling compared with the DMT analysis, as
certain ratios and terms that simply vanish in the DMT analysis are in
our case relevant and must be carefully handled.%%  One of the
%% interesting features of this work is the use of a method that tightly
%% bounds an arithmetic mean from both sides by (functions of) the
%% corresponding geometric mean, an interesting result originally due to
%% Specht~\cite{Specht:ZUR60} that has not received much attention in
%% information theory circles, but may prove useful in a variety of other
%% analyses as well.

We then analyze the {\em frequency-selective}, quasi-static MIMO
channel. Specifically we consider single carrier (SC) MMSE
equalization under zero-padding (ZP) and cyclic-prefix (CP)
transmission. SC-MMSE provides an attractive alternative to orthogonal
frequency division multiplexing (OFDM) due to its low complexity and
natural avoidance of the peak-to-average power ratio problem. The use
of cyclic prefix and zero padding has been investigated in the
literature, but the explicit tradeoff between the spectral efficiency
and diversity of MIMO SC-MMSE under these two schemes has been unknown
and is the subject of our work. We show that the diversity is a
function of number of antennas, channel memory and spectral
efficiency, and obtain the explicit tradeoff in the special case of
SIMO under CP transmission.

The results of this paper fully characterize the MIMO MMSE diversity
in the fixed rate flat quasi-static regime. We analyze both the cases
$N \ge M$ and $N<M$, showing that in either case it is possible for
the system to be limited to a diversity strictly less than $MN$. More
specifically, the central result of the paper is as follows: with $M$
transmit and $N$ receive antennas (for any $N$ and $M$) the diversity
is $d=\lceil \big( M 2^{-\frac{R}{M}}-(M-N)^+\big)^+\rceil^{2}+ |N-M|\lceil \big( M
2^{-\frac{R}{M}}-(M-N)^+\big)^+\rceil$, where $(\cdot)^+=\max(0,\cdot)$ and $\lceil
\cdot \rceil$ denotes rounding up to the next higher integer.  Our
results confirm and refine the earlier approximate results on the
diversity of MMSE MIMO receivers that were obtained for very high and
very low
rates~\cite{Hedayat:JSP07,Onggosanusi:ICAS02,Kumar:JIT09}. The MIMO MAC
channel is also studied.

Some of the related literature is as follows. The performance of MMSE
receiver in terms of reliability goes back to~\cite{Gao:JIT98} where
outage analysis was performed for MMSE SIMO diversity combiner in a
Rayleigh fading channel with multiple interferers.  In the context of
point-to-point MIMO systems, Gore et al.~\cite{Gore:ASI02} compared
the performance of MMSE D-BLAST with the ordered successive
cancellation V-BLAST.  They show that the former has better throughput
at low- and moderate SNR. Onggosanusi et al.~\cite{Onggosanusi:ICAS02}
studied MMSE and zero-forcing (ZF) MIMO receivers and noticed
their distinct outage performance at high-SNR, specifically for large
number of transmit antennas and low spectral efficiencies $R$, but
provided no analysis.

Hedayat and Nosratinia~\cite{Hedayat:JSP07} considered the outage
probability as a function of fixed rates $R$ under joint and separate
spatial encoding, but for MMSE they obtained results only in
the extremes of very high and very low rates. Kumar et
al.~\cite{Kumar:JIT09} provided a DMT analysis for the system
of~\cite{Hedayat:JSP07} and observed that the DMT analysis does not
predict the diversity of MMSE receivers at lower rates. We note that
all existing analyses are limited to the case where the number of
receive antennas ($N$) is greater than or equal the number of the
transmit antennas ($M$).

This paper is organized as follows. Section~\ref{sec:LinearRx}
describes the system model. Section~\ref{sec:OutageAnalysis} finds the
exponential order of outage. Section~\ref{sec:PEPAnalysis} bounds the
codeword error probabilities using the outage values, and derives the
final result. Section~\ref{sec:MAC} extends the result to the MAC
channel. Section~\ref{sec:freqSel} calculates the diversity of MIMO
MMSE reception in frequency-selective block-transmission
systems. Section~\ref{sec:SimRes} provides simulations that illuminate
our results.

%------------------------- Section 2: Linear Receivers ---------------------------

\section{Linear Receivers}
\label{sec:LinearRx}

\begin{figure*}
\centering
\includegraphics[width=5.5in]{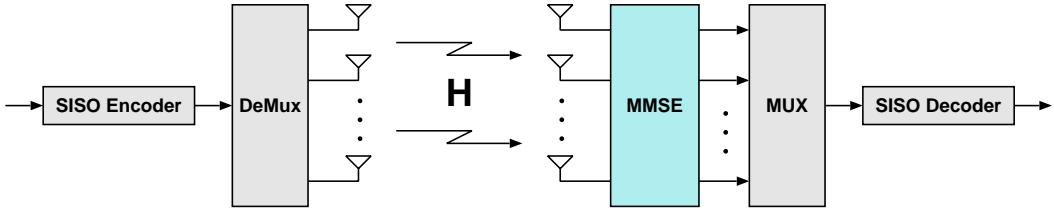}
\caption{MIMO system with linear MMSE receiver}
\label{fig:LinMIMOMMSE}
\end{figure*}

The input-output system model for flat fading MIMO channel with $M$ transmit and $N$
receive antennas is given by
\begin{equation}
 \by = \bH \bx + \bn 
\end{equation}
where $ {\bH \in \mathbb{C}^{N \times M}} $ is the channel
matrix whose entries are independent and identically distributed
complex Gaussian, $ {\bx \in \mathbb{C}^{\ M \times 1}}$ is the
transmitted vector, ${\bn \in \mathbb{C}^{\ N \times 1}} $ is the
Gaussian noise vector. The vectors ${\bx}$ and ${\bn}$ are
assumed independent.
We assume a quasi-static flat fading channel and perfect channel state
information (CSI) at the receiver (CSIR) and no CSI at the transmitter (CSIT), therefore transmit antennas operate with equal power.

We aim to characterize the diversity gain, $ d(R,M,N)$, as a function
of the spectral efficiency $R$ (bits/sec/Hz) and the number of
transmit and receive antennas. This requires a pairwise error
probability (PEP) analysis which is not directly tractable. Instead,
we find the exponential order of outage probability and then
demonstrate that outage and PEP exhibit identical exponential orders.

Following the notation of~\cite{Tajer:WCOM10}, we define the
outage-type quantities
\begin{align}
 P_{out}(R,N,M) &\triangleq  \prob(I(\bx;\by)<R)
\label {eq:OutageProb}\\
 d_{out}(R,N,M) &\triangleq  -\lim_{\rho\rightarrow\infty}\frac{\log P_{out}(R,M,N)}{\log \rho}
\end{align}
where $\rho$ is the per-stream signal-to-noise ratio (SNR).

We say that the two functions  $f(\rho)$ and $g(\rho)$ are {\em exponentially equal}, denoted by $f(p) \doteq g(p)$ when
\begin{equation}
\lim_{\rho \rightarrow \infty}
\frac{\log f(\rho)}{\log(\rho)}=\lim_{\rho \rightarrow \infty}
\frac{\log g(\rho)}{\log(\rho)} \nonumber
\end{equation}

The ordering operators $\dot{\leqslant}$ and $\dot{\geqslant}$ are also defined accordingly. If $f(\rho) \doteq \rho^{d}$, we say that $d$ is the {\it exponential order} of $f(p)$.

\subsection{MMSE Equalizer}
\label{subsec:Equalizer}

The equalizer, denoted by $\bW$, decouples the $M$ transmitted data
streams at the receiver (Figure~\ref{fig:LinMIMOMMSE}). The MMSE equalizer is obtained by minimizing
the mean square error (MSE) defined as $\expect
[||\bx-\bW^{H}\by||^2]$. It is usually
assumed~\cite{Hedayat:JSP07,Kumar:JIT09} that the number of transmit
antennas $M$ is no more than that of receive antennas $N$. In the
following, we start with $N\ge M$ but later generalize it to $N<M$ as well.

For $N \ge M$, using the orthogonality principle, the MMSE
equalizer is given by~\cite{Onggosanusi:ICAS02,Verdu98:book}
\begin{align}
 \bW &= \bH^{H}(\bH\bH^{H}+\rho^{-1}I)^{-1} \nonumber \\
         &= (\bH^{H}\bH+\rho^{-1}I)^{-1}\bH^{H} 
\end{align}

  The corresponding signal-to-interference and noise ratio (SINR) of the
  output stream $k$ of the MMSE detector is
\begin{align}
 \gamma_{k}&=\frac{1}{(\bI+\rho \bH^H\bH)^{-1}_{kk}}-1, \quad 1\leqslant k \leqslant M 
\label{eq:MMSESINR}
\end{align}
where $(\cdot)^H$ denotes matrix Hermitian, $(\cdot)^{-1}_{kk}$
denotes the diagonal element $k$ of the matrix inverse.

For the case $N<M$, it can be shown using a technique\footnote{In~\protect\cite{Gao:JIT98} an MMSE diversity combiner is used at the receiver in the presence of one transmit antenna and $M$ interferers.} very similar to~\cite[Appendix A]{Gao:JIT98} that the SINR expression \eqref{eq:MMSESINR} is again valid.

The square matrix $\mathcal{W}=\bH^H\bH$ is random, non-negative
definite, and obeys the {\em Wishart
  Distribution}~\cite{Telatar:EUR99,James:AMS64}. In this work,
the joint distribution of the eigenvalues of this
equivalent channel matrix opens the door to the development of our
analysis, as is also the case in many other MIMO results.

The equalizer output is
 \begin{equation}
 \by =\bW \bH \bx+\bW \bn. \label{eq:MMSEOUT}
\end{equation}
The signal streams of the transmit antennas may be either separately
or jointly encoded.  Separate encoding is simpler and has been fully
analyzed~\cite{Hedayat:JSP07}, but we mention the central result for
completeness.

\begin{theorem}[\cite{Hedayat:JSP07,Kumar:JIT09}]
\label{Theorem:SeparateEncoding}
In a MIMO system consisting of $M$ transmit and $N$ receive antennas ($N\geqslant M$),
under separate spatial encoding, the MMSE receiver achieves the diversity
\begin{equation}
 d_{out}(R,N,M)=N-M+1
\end{equation}
under either uniform or non-uniform rate assignment. 

Furthermore, it
has been established~\cite{Hedayat:JSP07,Kumar:JIT09} that the zero forcing
equalizer achieves diversity $N-M+1$ under both joint or separate spatial
encoding.
\end{theorem}

According to Theorem~\ref{Theorem:SeparateEncoding}, a MMSE receiver
operating under separate spatial encoding (e.g. horizontal encoding V-BLAST)
will have no more diversity gain than ZF receiver.

%---------------------- Section 3: Outage Analysis  ---------------------
\section{Outage Analysis}
\label{sec:OutageAnalysis}

 We now consider the MMSE diversity where the data stream
is first encoded then multiplexed into $M$ sub-streams, each
transmitted by one antenna. This approach is known to improve the
performance compared with separate coding of the
streams~\cite{Tse05:book}. Outage occurs if the channel fails to
support the target rate \cite{Telatar:EUR99}. After channel
equalization, the $M$ sub-streams $x_k$ are decoupled and thus the
mutual information between the transmitted vector $\bx$ and the
received vector $\by$ given CSIR is~\cite{Onggosanusi:ICAS02}
\begin{equation}
 I(\bx,\by)=\sum^M_{k=1} I(x_k,y_k)
\label {eq:MutualInfo}
\end{equation}
Thus from \eqref{eq:OutageProb} and \eqref{eq:MutualInfo}, $ P_{out}$ is given by
\begin{equation}
P_{out} = \prob\bigg(\sum^M_{k=1}\log(1+\gamma_{k})<R\bigg)  
\label{eq:MMSEoutage}
\end{equation}
Substituting MMSE SINR $(\gamma_{k})$ from \eqref{eq:MMSESINR} in \eqref{eq:MMSEoutage} we get
\begin{equation}
P_{out} = \prob\bigg(\sum^M_{k=1}\log(\bI+\rho \mathcal{W})^{-1}_{kk}>-R\bigg)
\label{eq:UpperBoundSum}
\end{equation}

The dependence on the diagonal elements of the random matrix
$(\bI+\rho \mathcal{W})^{-1}_{kk}$ makes further analysis
intractable. We instead proceed to provide lower and upper bounds on
the outage probability. In Section~\ref{sec:PEPAnalysis} we will show that outage
probability ($P_{out}$) and pairwise error probability (PEP) exhibit
identical exponential error.

%---------------------------------------------------------------------------------
\subsection{Outage Upper Bound} 
\label{sec:OutageUB}

\begin{lemma} 
\label{lemma:outageUB}
 For an MMSE MIMO system consisting of $M$ transmit and $N$ receive
 antennas, under quasi-static Rayleigh fading, we have
$P_{out}(R,M,N) \dot{\leqslant} \rho^{-d_{out}(R,M,N)}$ where
\begin{align}
d_{out}(R,M,N)&=\bigg\lceil
 \big( M2^{-\frac{R}{M}}-(M-N)^+ \big)^+ \bigg\rceil^2+\nonumber\\
&\hspace{10 pt}\big|N-M\big|\bigg\lceil
 \big( M2^{-\frac{R}{M}}-(M-N)^+ \big)^+\bigg\rceil. \label{eq:d_out}
\end{align}
where $()^+$ denotes the $\max(0,\cdot)$.
\end{lemma}

\begin{Proof}

We begin by bounding the sum in \eqref{eq:UpperBoundSum} via Jensen's inequality
\begin{align}
\sum^M_{k=1}\log \big(\bI+\rho \mathcal{W}\big)^{-1}_{kk} &\leq M
\log\big(\sum^M_{k=1}\frac{1}{M}(\bI+\rho \mathcal{W})^{-1}_{kk}
\big)\nonumber\\ &= M \log \big(\frac{1}{M}\mathrm{tr} \big(
(\bI+\rho \mathcal{W})^{-1} \big) \big)\nonumber  \\ &= M \log
\big(\frac{1}{M} \sum^M_{k=1}\frac{1}{1+\rho \lambda_k}
\big) \label{eq:MutualInfoLambdas}
\end{align}
where~\eqref{eq:MutualInfoLambdas} is true because trace is equal to
the sum of eigenvalues.

Notice that for $N<M$ only $N$ eigenvalues are
non-zero. hence~\eqref{eq:MutualInfoLambdas} can be written as
\begin{equation}
M \log \bigg(\frac{1}{M} \sum^L_{k=1}  \frac{1}{1+\rho \lambda_k}+(M-N)^+
\bigg) \label{eq:MutualInfoLambdasMin}
\end{equation}
where $L=\min(M,N)$.

  Substituting \eqref{eq:MutualInfoLambdasMin} in
\eqref{eq:UpperBoundSum}, we have

\begin{equation}
P_{out} \leq \prob \bigg( \sum^L_{k=1}\frac{1}{1+\rho \lambda_k}
\geqslant M2^{-\frac{R}{M}}-(M-N)^+ \bigg)
\label{eq:OutageWithLambdas}
\end{equation}

Define:
 \begin{equation}
\alpha_k \triangleq - \frac{\log \lambda_k}{\log \rho},\quad
\mathrm{for}\quad k=1,...,n\quad,
\label{eq:alpha}
\end{equation}
based on which we can write the exponential equality 
\begin{equation}
\frac{1}{1+\rho \lambda_k} \doteq
 \begin{cases}
 \rho^{\alpha_k-1} &\mbox{$\alpha_k<1$ } \\
 1  &\mbox{$\alpha_k>1$}
\end{cases}
\label{eq:AsymptoticTerm}
\end{equation}

Define $\balpha $ $=[\alpha_1,...,\alpha_n]$ and a new random variable 
 \begin{equation}
M(\boldsymbol {\alpha)} \triangleq \sum_{\alpha_k>1}1
\end{equation}
This definition is based on the observation that the term
$\frac{1}{1+\rho \lambda_k}$ defined in \eqref{eq:AsymptoticTerm} is either zero
or one at high SNR, therefore to characterize $\sum_k\frac{1}{1+\rho
  \lambda_k}$ at high SNR we count the ones. Thus
 \begin{align}
\sum^n_{k=1}\frac{1}{1+\rho \lambda_k} &\doteq
\sum_{\alpha_k>1}1+\sum_{\alpha_k<1}\rho^{\alpha_k-1} \\ &
\doteq
M(\balpha)+\max_{\{\alpha_k:\alpha_k<1\}}{\rho^{\alpha_k-1}} 
\end{align}
$M(\balpha)$ inherits its randomness from
 $\lambda_1,\ldots,\lambda_n$. The bound
 in~\eqref{eq:OutageWithLambdas} is evaluated by computing the
 probability of $\{\balpha \in \mathcal{A}\}$, where $ \mathcal{A} =
 \{\alpha:
 M(\balpha)+\max_{\{\alpha_k:\alpha_k<1\}}{\rho^{\alpha_k-1}}>M2^{-\frac{R}{M}}-(M-N)^+
 \} $ denotes the outage event based on the approximation
 in~\eqref{eq:OutageWithLambdas}. In order to evaluate the probability
 of this event we need the joint distribution of the eigenvalues, or equivalently the distribution of $\balpha$. The distribution follows Wishart distribution and was initially discovered by~\cite{James:AMS64} . The distribution of $\balpha$ can be easily evaluated as follows~\cite{Zheng:JIT03}. 

Let $\bR$ be an $m \times n$ ($m \geqslant n$) random matrix whose entries are $\mathcal{CN}(0,1)$. The joint PDF of the ordered random variables $\balpha$ (defined in \eqref{eq:alpha} for the eigenvalues of $\bR^{H}\bR$) is given by
 \begin{align}
\prob(\balpha) =K^{-1}_{m,n}&(\log \rho)^{n}
\prod^n_{i=1} \rho^{-(m-n+1)\alpha_i}  \times \nonumber\\
&\prod_{i<j}
|\rho^{-\alpha_i}-\rho^{-\alpha_j}|^2 \exp \bigg[ -\sum_{i=1}^n
  \rho^{-\alpha_i} \bigg]
\label{eq:WishartJointDist}
\end{align}
where $K^{-1}_{m,n}$ is a normalizing factor.

Using the distribution of $\balpha$ for the defined matrix ${\mathbf R}$, the asymptotic outage bound is 
\begin{align}
&P_{out}  \dot{\leqslant} \int \limits_{\mathcal{A}} \prob(\balpha) d\balpha
\nonumber
 \\ &= K^{-1}_{m,n}(\log \rho)^{n} \int \limits_{\mathcal
  {A}} \prod^n_{i=1} \rho^{-(m-n+1)\alpha_i}  \prod_{i<j} |\rho^{-\alpha_i}-\rho^{-\alpha_j}|^2 \times\nonumber\\
&\hspace{80pt} \exp \bigg[
  -\sum_{i=1}^n \rho^{-\alpha_i} \bigg] \;\text{d} \balpha
\end{align}

The simplification of the integral follows from~\cite{Zheng:JIT03}. The term outside the integral has no
effect on the exponent. The term $
|\rho^{-\alpha_i}-\rho^{-\alpha_j}|$ is dominated by
$\rho^{-\alpha_i}$ at high SNR. We now divide the integration range into
$\mathcal{A'} =\mathcal{A} \cap \real^{n}_{+} $ and its complement. If
$\balpha \notin \mathcal{A'}$, the exponential term will dominate the
other terms and will drive the integral to zero. If $\balpha \in
\mathcal{A'}$, the exponential term is approximately 1 at high SNR and
will disappear. Therefore
\begin{align}
P_{out} &\;\dot{\leqslant}\; \int \limits_{\mathcal {A'}} \prod^n_{i=1}
\rho^{-(m-n+1)\alpha_i} \prod_{i<j}
|\rho^{-\alpha_i}-\rho^{-\alpha_j}|^2 \;\text{d}\balpha \nonumber \\ &\doteq
\int \limits_{\mathcal {A'}} \prod^n_{i=1} \rho^{-(2i-1+m-n)\alpha_i}
\;\text{d}\balpha \label{eq:OutageIntegral}
\end{align}
where
\begin {align}
 \mathcal {A'} &= \{M(\balpha)>M2^{-\frac{R}{M}}-(M-N)^+ \} \nonumber \\
              &= \{ \alpha_1>1,...,\alpha_{S}>1,\alpha_{S+1}>0,...\alpha_L>0 \} 
\end {align}
and $S=\big\lceil \big(M2^{-\frac{R}{M}}-(M-N)^+\big)^+\big\rceil$. The integration region
$\mathcal{A'}$ has boundaries that are parallel to nonnegative orthant
$\mathbb{R}_+^n$, therefore the integration over multiple variables
in~\eqref{eq:OutageIntegral} can be separated:
\begin {align}
P_{out} \quad &\dot{\leqslant}\quad \prod^n_{i=1}\int \limits_{\mathcal {A'}}  \rho^{-(2i-1+m-n)\alpha_i}
\;\text{d}\balpha \\  &= \quad \rho^{-\sum_{i=1}^{S} (2i-1+m-n)} \nonumber \\
      &= \quad \rho^{-(S^2+(m-n)S)},  \quad \mbox{for $m\geqslant n$} \\
 &= \quad \rho^{-(S^2+|m-n|S)}, \quad \mbox{for general $m, n$} \label{eq:DivRes}\\
&= \quad \rho^{-d_{out}} \nonumber
\end {align}

which establishes the proof of Lemma~\ref{lemma:outageUB}.
\end{Proof}

%-----------------------------------------------------------------------
\subsection{Outage Lower Bound} 
\label{sec:LB}

\begin{lemma} 
\label{lemma:outageLB}
 For an MMSE MIMO system consisting of $M$ transmit and $N$ receive
 antennas (and $L=\min\{M,N\}$), operating under quasi-static Rayleigh
 fading, we have $P_{out}(R,M,N)
 \dot{\geqslant} \rho^{-d_{out}(R,M,N)}$ where
\begin{align}
d_{out}(R,M,N)&=\bigg\lceil
 \big( M2^{-\frac{R}{M}}-(M-N)^+ \big)^+ \bigg\rceil^2+\nonumber\\
&\hspace{10 pt}\big|N-M\big|\bigg\lceil
 \big( M2^{-\frac{R}{M}}-(M-N)^+ \big)^+\bigg\rceil. \nonumber
\end{align}
\end{lemma}

 \begin{Proof}
The lower bound is also based on Jensen's inequality. Recall
\begin{align}
P_{out} &= \prob\bigg(\sum^M_{k=1}\log(1+\gamma_{k})<R\bigg)  \nonumber\\
      &= \prob\bigg( \sum^M_{k=1}\log \frac{1}{(\bI+\rho \mathcal{W})^{-1}_{kk}}< R\bigg) \nonumber\\
 &\geqslant \prob\bigg( M \log \frac{1}{M}\sum^M_{k=1}\frac{1}{(\bI+\rho \mathcal{W})^{-1}_{kk}}< R\bigg) \label{eq:OutJenLower}
\end{align}
Let the eigen decomposition of $\bH^H\bH$ be given by
$\bH^H\bH=\bU^H\Lambda\bU$ where $\bU$ is unitary and $\Lambda$ is a
diagonal matrix that has the eigenvalues of the Wishart matrix
$\mathcal{W}$ on its diagonal. Let the vector $\bu_k$ be the column
$k$ of the matrix $\bU$ and $u_{\ell k}$ be the element $\ell$ of this
column, we have
\begin{align}
(\bI+\rho \mathcal{W})^{-1}_{kk} &= \bu^H_k (\bI+\rho \Lambda)^{-1} \bu_k\nonumber\\
&= \sum_{\ell=1}^{M} \frac{|u_{\ell k}|^2}{1+\rho\lambda_{\ell}}\nonumber\\ &\triangleq  S_k. \label{eq:EqKunmar}
\end{align}
Let $\bar{k}=\arg\min_k S_k$.  Using~\eqref{eq:EqKunmar}, we can bound
the sum in~\eqref{eq:OutJenLower}
\begin{align}
\frac{1}{M} \sum_{k=1}^{M} \frac{1}{(\bI+\rho \mathcal{W})^{-1}_{kk}}  &= \frac{1}{M} \sum_{k=1}^{M}\frac{1}{S_k} \nonumber\\
  &\leqslant \;  \frac{1}{\min_k S_k} \label{eq:EqKunmar2} \\
&= \; \frac{1}{ S_{\bar{k}}}
\end{align}
thus the outage bound in~\eqref{eq:OutJenLower} can be further bounded
using~\eqref{eq:EqKunmar2}
\begin{align}
P_{out}  &\geqslant \prob\bigg( M \log \frac{1}{M}\sum^M_{k=1}\frac{1}{(\bI+\rho \mathcal{W})^{-1}_{kk}}< R\bigg) \nonumber\\
&\geqslant \prob\bigg( M \log \frac{1}{ S_{\bar{k}}}< R\bigg) \nonumber\\
&=  \prob\bigg( S_{\bar{k}} >  2^{-\frac{R}{M}} \bigg)   \label{eq:OutJenLower2}
\end{align}

We now bound~\eqref{eq:OutJenLower2} by conditioning on the event
$\mathcal{B}\triangleq \big\{ |u_{\ell \bar{k}}|^2 \geqslant
\frac{a}{M} \big\}$ where $a$ is a positive real number that is
slightly smaller than one, i.e. $a=1-\epsilon$, and $\epsilon$ is a small
positive number. We then have
\begin{align}
\prob\bigg( S_{\bar{k}} &> 2^{-\frac{R}{M}} \bigg) \geqslant
\prob\bigg( S_{\bar{k}} > 2^{-\frac{R}{M}} \big|\mathcal{B} \bigg)
\prob(\mathcal{B})\nonumber\\
 &=\prob\bigg( \sum_{\ell=1}^{M} \frac{|u_{\ell \bar{k}}|^2}{1+\rho\lambda_{\ell}} > 2^{-\frac{R}{M}} \bigg|\mathcal{B} \bigg)
\prob(\mathcal{B}) \nonumber\\
 &\geqslant\prob\bigg( \frac{1}{M}\sum_{\ell=1}^{M} \frac{a}{1+\rho\lambda_{\ell}} > 2^{-\frac{R}{M}}  \bigg) \prob(\mathcal{B}) \nonumber\\
&\doteq \prob\bigg( \frac{1}{M}\sum_{\ell=1}^{M} \frac{a}{1+\rho\lambda_{\ell}} > 2^{-\frac{R}{M}}  \bigg) \label{eq:OutJenLower4}\\
&= \prob\bigg(\sum_{\ell=1}^{M} \frac{1}{1+\rho\lambda_{\ell}} > \frac{M}{a}2^{-\frac{R}{M}}  \bigg) \nonumber\\
&= \prob \bigg(  \frac{1}{M} \sum^L_{\ell=1}\frac{1}{1+\rho \lambda_{\ell}} > \frac{M}{a}2^{-\frac{R}{M}}-(M-N)^+ \bigg) \label{eq:Bou}
\end{align}
where~\eqref{eq:OutJenLower4} follows because $\prob(\mathcal{B})$ is
finite and independent of $\rho$; this can be proved similarly
to~\cite[Appendix A]{Kumar:JIT09}. To make the upcoming expressions
compact, we introduce a new variabe $\kappa\triangleq
\frac{M}{a}2^{-\frac{R}{M}}-(M-N)^+$
\begin{align}
\prob \bigg(  \frac{1}{M} \sum^L_{\ell=1}\frac{1}{1+\rho \lambda_{\ell}} > \kappa \bigg) \label{eq:Bou2}
\end{align}

Whenever $M2^{-\frac{R}{M}}$ is non-integer, the constant $a$ can be
chosen such that $\big\lceil \big( M2^{-\frac{R}{M}}-(M-N)^+ \big)^+ \big\rceil = \big\lceil \big(\frac{M}{a}2^{-\frac{R}{M}}-(M-N)^+\big)^+ \big\rceil$. We note this is satisfied for all
rates, with the exception of an isolated set of points. As long as
$M2^{\frac{R}{M}} \notin {\mathbb N}$ we have:
\begin{align}
P_{out } &\geqslant  \prob \bigg(
\sum^L_{\ell =1}\frac{1}{1+\rho \lambda_{\ell}} > \kappa
\bigg) \nonumber\\
 &\doteq \prob \bigg(
\sum^L_{k=1}\frac{1}{1+\rho \lambda_{\ell}} > \lceil\kappa \rceil
\bigg) 
\label{eq:OutageLowerBound}
\end{align}
The remaining steps follow similarly to the proof of
Lemma~\ref{lemma:outageUB}. Thus $P_{out}
\;\dot{\geqslant}\;\rho^{-d_{out}}$ with $d_{out}$ is given by
Lemma~\ref{lemma:outageLB}.

On the set of isolated points $M2^{-\frac{R}{M}} \in {\mathbb N}$, the
right hand side of Eq.~\eqref{eq:OutageLowerBound} obeys a slightly
weaker upper bound by replacing $\kappa $ with $\kappa +1$. We can combine the cases where $M2^{-\frac{R}{M}}$ is integer and non-integer to write the upper bound compactly as follows:
\begin{align*}
d_{out}(R,M,N)& \le\bigg\lfloor
 \big( M2^{-\frac{R}{M}}+1-(M-N)^+ \big)^+ \bigg\rfloor^2+\nonumber\\
&\hspace{10 pt}\big|N-M\big|\bigg\lfloor
 \big( M2^{-\frac{R}{M}}+1-(M-N)^+ \big)^+\bigg\rfloor. \nonumber
\end{align*}
Inspection shows that this bound is tight against the lower bound
everywhere except its discontinuity points. In other words, the upper bound is left-continuous while the  lower
bound was right-continuous at
the discontinuity points.

 \end{Proof}

%---------------------- Section: PEP Analysis  -----------------------------------
\section{PEP Analysis}
\label{sec:PEPAnalysis}

Recalling that the diversity is roughly defined as the slope of PEP at
high SNR, we now proceed to bound the PEP tightly from both sides
using the outage results already obtained.

%---------------------------------------------------------------------------------
\subsection{PEP Upper Bound} 
\label{sub:PEPUB}

We start by a lower bound that is inspired by~\cite[Lemma
  5]{Zheng:JIT03} but requires a more careful treatment since we are
analyzing rate, not the DMT (see the Introduction).

%--------------
% Lemma
%--------------
\begin{lemma} \label{lemma:PEPUB}
For a quasi-static fading MIMO channel with MMSE receiver we have
$d_{out}(R,M,N) \geqslant d(R,M,N)$.
\end{lemma}

\begin{Proof} 
\label{proof lemma:PEPUB}

Denote $E$ for an error event, and let $x \in \mathcal{C}$ be the
transmitted codeword from a codebook $\mathcal {C}$ of size $2^{Rl}$
where $R$ and $l$ are code rate and code length respectively. Define $\bfunc=\bW
\bH \bx$ that accounts for the combined effect of channel and equalizer. The transmit messages
are assumed equi-probable so the entropy
$\mathcal{H}=\log{|\mathcal{C}|}=Rl$. Applying the Fano inequality~\cite{Cover06:book}
\begin{equation}
\prob(E|\bfunc=f) \geqslant
\frac{Rl-I(\bx;\by|\bfunc=f)}{Rl} -
\frac{\mathcal{H}(\prob(E)|\bfunc=f)}{Rl} \label{eq:Fano}
\end{equation}

By defining $\mathcal{D}_{\delta}$ for any $\delta>0$
as
$
\mathcal{D}_{\delta} \triangleq
\{f:I(\bx;\by|\boldsymbol{f}=f)<l(R-\delta) \},
$
and noting that $\mathcal{H}(\prob(E)|f \in \mathcal{D}_{\delta})
\leqslant \mathcal{H}(\prob(E))$ from~\eqref{eq:Fano}, we get 
\begin{align}
\prob(E|\bfunc\in \mathcal{D}_{\delta}) &\geqslant
\frac{Rl-I(\bx;\by|\bfunc\in
  \mathcal{D}_{\delta})}{Rl} - \frac{\mathcal{H}(\prob(E))}{Rl}
\nonumber \\ &\geqslant \frac{\delta}{R}-
\frac{\mathcal{H}(\prob(E))}{Rl}. \label{eq:PEPFano}
\end{align}
Also by using the definition of $P_{out}$ we have
\begin{equation}
\prob( \bfunc\in \mathcal{D}_{\delta}) =
\prob \big(I(\bx;\by)<l(R-\delta) \big) \doteq
\rho^{-d_{out}(R-\delta,M,N)} \label{eq:PEPResult}
\end{equation}
For small enough values of $\delta>0$, we have
$d_{out}(R,M,N)=d_{out}(R-\delta,M,N)$ since $d_{out}(R,M,N)$ is
left-continuous with respect to $R$. Hence, by invoking \eqref{eq:PEPFano} and \eqref{eq:PEPResult}, the error probability is given by
\begin{align}
\prob_{\text{err}}(R,M,N) &= \prob(E|f \in \mathcal{D}_{\delta})
\prob(f \in \mathcal{D}_{\delta})+ \nonumber \\ &\quad \prob(E|f\notin
\mathcal{D}_{\delta}) \prob(f\notin
\mathcal{D}_{\delta}) \nonumber\\
&\geqslant \prob(E|f \in \mathcal{D}_{\delta}) \prob(f \in
  \mathcal{D}_{\delta}) \nonumber\\
 & \dot{\geqslant} \bigg(
\frac{\delta}{R}- \frac{\mathcal{H}(\prob(E))}{Rl} \bigg)
\rho^{-d_{out}} \nonumber \\ & \doteq \rho^{-d_{out}} 
\end{align}
where we have used  $\big(
\frac{\delta}{R}- \frac{\mathcal{H}(\prob(E))}{Rl} \big) \doteq 1$,
which was derived in~\cite{Tajer:WCOM10}. 
This establishes the proof of the PEP upper bound.
\end{Proof}

%---------------------------------------------------------------------------------
\subsection{PEP Lower Bound} 
\label{sec:PEPLB}

We begin by writing the error probability in terms of error event $E$
and outage event $O$
\begin{equation}
\prob_{\text{err}}(R,M,N) = \prob(E|O) \cdot P_{out} + \prob(E,\bar{O}) \nonumber \\
\end{equation}

In Section \ref{sec:OutageUB} we have shown that, based on the event
 $ \big \{
\sum^L_{k=1}\frac{1}{1+\rho \lambda_k} \geqslant M2^{-\frac{R}{M}}-(M-N)^+
\big\}$, the outage probability is upper bounded by $P_{out}
\dot{\leqslant} \rho^{-d_{out}}$.  Hence, the error probability can be
bounded as
\begin{align}
\prob_{\text{err}}(R,M,N) &\dot{\leqslant} \prob(E|
O) \, \rho^{-d_{out}} + \prob(E,\bar{O}) \nonumber \\ &
\leqslant \rho^{-d_{out}} + \prob(E,\bar{O}) \label{eq:Perr}
\end{align}
We intend to show that $\rho^{-d_{out}} \;
\dot{\geqslant} \hspace{0.05in} \prob(E,\bar{O})$, and thus
$\prob_{\text{err}}(R,M,N) \; \dot{\leqslant} \;
\rho^{-d_{out}}$ which produces the following lemma.

\begin{lemma} \label{lemma:PEPLB}
For a quasi-static  fading MIMO channel with MMSE receiver we have
$d_{out}(R,M,N) \leqslant d(R,M,N)$.
\end{lemma}

%---------------------------------------------------------------------
\begin{Proof} 
\label{proof lemma:PEPLB}

We begin by giving a sketch of the proof then we proceed with the
details. The first part of the proof consists of developing a bound on
PEP conditioned on $H$, namely $P[s_k\rightarrow s_j|\bH=H]$. To do
this we obtain an upper bound of the variance of the SINR which is
expressed in terms of the eigenvalues of the Wishart matrix
$\mathcal{W}$, resulting in $ \prob[E | \bH=H] \leqslant 4\exp({-\big(
  \sum^{L}_{k=1} \frac{\rho \lambda_k}{(1+\rho\lambda_k)^2}
  \big)^{-1}}) \, .  $ The PEP is used to derive a conditional union
bound on error. We then divide the channel events into two sets based
on the exponential order of the eigenvalues: the set where
$M(\balpha)=0$ and otherwise.  We apply Bayes theorem on the union
bound using these two sets. The calculation of the terms of the
Bayesian gives $\prob(E,\bar{O}) \dot{\leqslant} \rho^{-MN}\leqslant
\rho^{-d_{out}}$ as desired.

We now proceed in detail. We want to compute the probability that the
transmitted symbol $x(k)=s_l$ is erroneously detected as $x(k)=s_j$.

Recalling the equalizer output given by~\eqref{eq:MMSEOUT}, define the noise-plus-interference signal
\begin{equation}
  \tilde{\bn} =
\by-\sqrt{\rho}\bx =
\sqrt{\rho}(\bW \bH-\bI)\bx+\bW \bn
\end{equation}
 
Using the eigen-decomposition of $H$ and noting that $E(\bn)=0$
  and $E(\bn\bn^{H})=I$, we have 
\begin{align}
\mu_{\tilde{\bn}} &\triangleq E(\tilde{\bn})=\sqrt{\rho}(\bW\bH-\bI) =
-\rho^{\frac{1}{2}} (\mathcal{W}+\rho^{-1} \bI)^{-1}
\bx \label{eq:mu}\\ \boldsymbol{R_{\tilde{\bn}}} &\triangleq
E(\tilde{\bn} \tilde{\bn}^H) = (\mathcal{W}+\rho^{-1} \bI)^{-1}
\label{eq:R}
\end{align}
Thus the variance of the noise sample $\tilde{n} (k)$ is given by
\begin{align}
\sigma^2_{\tilde{\bn}}(k) &=
\boldsymbol{R_{\tilde{\bn}}}(k,k) -
|\mu_{\tilde{\bn}}(k)|^2 \nonumber \\ &=
(\mathcal{W}+\rho^{-1} \bI)^{-1}_{kk} -\rho^{-1}
(\mathcal{W}+\rho^{-1} \bI)^{-2 }_{kk} \label{eq:variance}
\end{align}

where $|\mu_{\tilde{\bn}}(k)|^2$ is the $k^{th}$ diagonal
of the matrix $E(\tilde{\bn})E(\tilde{\bn}^H)$
and $k$ counts from 1 to $M$.

By defining $e_{jl} \triangleq \frac{s_j-s_l}{|s_j-s_l|}$, the
probability of erroneous detection for channel realization is given by
\begin{equation}
\begin{split}
\prob[s_l \rightarrow & s_j | \bH=H] \\ &= \prob \bigg[ \frac{\rho}{4}
  |s_j-s_l|^2 \leqslant |e^*_{jl}(y(k)-\sqrt{\rho}s_l)|^2
  \bigg |\bH=H \bigg] \\ &\leqslant \prob \bigg[ \frac{\rho}{4}
  |s_j-s_l|^2 \leqslant |\tilde{n}_k|^2 \bigg | \bH=H \bigg] 
\end{split}
\end{equation}
where the inequality holds since $|e^*_{jl}(y(k)-\sqrt{\rho}s_l)| \leqslant |e^*_{jl}||(y(k)-\sqrt{\rho}s_l)|=|(y(k)-\sqrt{\rho}s_j)|=|\tilde{n}(k)|$. 

Denoting the real and imaginary parts of $\tilde{n}(k)$ by $\tilde{n}_r(k)\sim \mathcal{N}(\mu_r(k),\sigma^2_r(k))$ and  $\tilde{n}_i(k)\sim \mathcal{N}(\mu_i(k),\sigma^2_i(k))$ respectively, we then have
 \begin{align}
 &\big\{ \frac{\rho}{4}|s_j-s_l|^2 \leqslant |\tilde{n}(k)|^2 \big\} \nonumber \\
 &\subset \big\{ \frac{\rho}{16}|s_j-s_l|^2 \leqslant |\tilde{n}_r(k)|^2 \big\} \cup  \big\{ \frac{\rho}{16}|s_j-s_l|^2 \leqslant |\tilde{n}_i(k)|^2 \big\}
 \end{align}

Applying the property of the Gaussian tail function $Q(x) \leqslant e^{(-x^2/2)}$ for the pairwise error probability, we obtain
\begin{align}
&\prob[s_k \rightarrow  s_j | \bH=H] \nonumber\\
  &\leqslant e^{\bigg(
  -\frac{(\frac{\sqrt{\rho}}{4}|s_j-s_l|-\mu_r(k))^2}{\sigma_{r}^2(k)}
  \bigg)} + e^{\bigg(
  -\frac{(\frac{\sqrt{\rho}}{4}|s_j-s_l|+\mu_r(k))^2}{\sigma_{r}^2(k)}
  \bigg)} \nonumber
\\ &+  e^{\bigg(
  -\frac{(\frac{\sqrt{\rho}}{4}|s_j-s_l|-\mu_i(k))^2}{\sigma_{i}^2(k)}
  \bigg)} + e^{\bigg(
  -\frac{(\frac{\sqrt{\rho}}{4}|s_j-s_l|+\mu_i(k))^2}{\sigma_{i}^2(k)}
  \bigg)}  \nonumber  \\
 & \leqslant e^{\bigg(
  -\frac{(\frac{\sqrt{\rho}}{4}|s_j-s_l|-\mu_r(k))^2}{\sigma_{\tilde{\bn}}^2(k)}
  \bigg)} + e^{\bigg(
  -\frac{(\frac{\sqrt{\rho}}{4}|s_j-s_l|+\mu_r(k))^2}{\sigma_{\tilde{\bn}}^2(k)}
  \bigg)} 
\nonumber\\ &+  e^{\bigg(
  -\frac{(\frac{\sqrt{\rho}}{4}|s_j-s_l|-\mu_i(k))^2}{\sigma_{\tilde{\bn}}^2(k)}
  \bigg)} + e^{\bigg(
  -\frac{(\frac{\sqrt{\rho}}{4}|s_j-s_l|+\mu_i(k))^2}{\sigma_{\tilde{\bn}}^2(k)}
  \bigg)} 
 \label{eq:PEP}
\end{align}
where the last step holds as $\sigma^2_{\bn}(k)=\sigma^2_r(k)+\sigma^2_i(k) \geqslant \sigma^2_r(k), \sigma^2_i(k)$.

Now we proceed by showing that $\mu_i(k) \dot{\leqslant} \rho^{\frac{1}{2}}$. Consider the eigen decomposition of
\begin{align}
[\mathcal{W} + \rho^{-1} \bI]^{-1} &= \mathbf{U}^H [\Lambda+\rho^{-1} \bI]^{-1} \mathbf{U} \nonumber \\
&= \mathbf{U}^H \big[ \diag\big\{ \frac{1}{\lambda_k+\rho^{-1}}  \big\} \big] \mathbf{U}
\end{align}
where $\mathbf{U}$ is unitary matrix, and $\Lambda$ is the eigen decomposition of $\mathcal{W}$.
Note that $\lambda_k+\rho^{-1} \dot{\geqslant} \rho^{-1}$ or $\frac{1}{\lambda_k+\rho^{-1}} \dot{\leqslant} \rho$. Therefore, all elements of the matrix $\pm \mathbf{U}^H [\Lambda+\rho^{-1} \bI]^{-1} \mathbf{U}$, being linear combination of $\{ \frac{1}{\lambda_k+\rho^{-1}} \}$, cannot grow faster than $O(\rho)$, and thus the elements of  $\pm \rho^{\frac{1}{2}} [\mathcal{W}+\rho^{-1} \bI]^{-1} $ cannot grow faster than $O(\rho^{\frac{1}{2}})$, i.e. $\pm \mu_{\tilde{\bn}}(k) \dot{\leqslant{\rho^{\frac{1}{2}}}}$ and therefore $ \rho^{\frac{1}{2}} \pm\mu_{\tilde{\bn}}(k) \doteq \rho^{\frac{1}{2}}$. The same result holds for $\mu_r(k)$ and $\mu_i(k)$.

As a result, for any $s_j$ and $s_l$, $\frac{\sqrt{\rho}}{4}|s_j-s_l| \pm \mu_r(k) \doteq \rho^{\frac{1}{2}} \pm \mu_r(k) \doteq \rho^{\frac{1}{2}}$ and similarly $\frac{\sqrt{\rho}}{4}|s_j-s_l| \pm \mu_i(k) \doteq \rho^{\frac{1}{2}}$.
Thus from \eqref{eq:PEP}, we have 
\begin{equation}
\prob[s_k \rightarrow  s_j | \bH=H] \quad \dot{\leqslant}\quad 4 e^{-\frac{\rho}{\sigma^2_{\tilde{\bn}}(k)}} \label{eq:PEPAsym}
\end{equation}

 Now we bound the variance in \eqref{eq:variance} and apply it in \eqref{eq:PEPAsym}
\begin{align}
\sigma^2_{\tilde{\bn}}(k) &\leqslant  \sum^{L}_{k=1} \bigg[
(\mathcal{W}+\rho^{-1} \bI)^{-1}_{kk} -\rho^{-1}
(\mathcal{W}+\rho^{-1} \bI)^{-2 }_{kk} \bigg]\nonumber \\ &=
\sum^{L}_{k=1} \bigg[\frac{\rho}{1+\rho\lambda_k} -
\frac{\rho}{(1+\rho\lambda_k)^2}\bigg] = \sum^{L}_{k=1}
\frac{\rho^2 \lambda_k}{(1+\rho\lambda_k)^2} \label{eq:varBound}
\end{align}

Denoting the error event $E$ and using \eqref{eq:varBound}, the probability of erroneous detection in \eqref{eq:PEPAsym} is bounded as 
\begin{align}
 \prob[E | \bH=H]
%% &\dot{\leqslant}
%% 4e^{\big(
%%   -\frac{\rho}{\sigma^2_{\boldsymbol{\tilde{n}}}(k)} \big)} \nonumber
  \leqslant 4e^{-\big( \sum^{L}_{k=1} \frac{\rho
    \lambda_k}{(1+\rho\lambda_k)^2} \big)^{-1}}
\end{align}

Applying the union bound, we get
\begin{equation}
\prob(E|\bH=H) \dot{\leqslant} 2^{Rl} e^{-\big( \sum^{L}_{k=1} \frac{\rho
    \lambda_k}{(1+\rho\lambda_k)^2} \big)^{-1}} \label{eq:PEPBound}
\end{equation}

Based on \eqref{eq:PEPBound}, we can evaluate $P(E,\bar{O })$ in \eqref{eq:Perr} as
follows. Recalling the exponential inequality
 \begin{align}
\sum^n_{k=1}\frac{1}{1+\rho \lambda_k} &\doteq
\sum_{\alpha_k>1}1+\sum_{\alpha_k<1}\rho^{\alpha_k-1} \\ &
\doteq
M(\balpha)+\max_{\{\alpha_k:\alpha_k<1\}}{\rho^{\alpha_k-1}} \nonumber \\
&\doteq M(\balpha) \label{eq:M}
\end{align}

Consider the two regions: $\{\balpha: M(\balpha)=0 \}$ and $\{\balpha:
M(\balpha) \geqslant 1 \}$. At high SNR the event $\bar{O}$ is
equivalent to $\{ \balpha : M(\balpha) \dot{\leqslant} \lceil
M2^{-\frac{R}{M}}-(M-N)^+\rceil \}$.

In the first region $\{M(\balpha)=0 \}$, at any rate $R \geqslant 0$
we have $\{\balpha : \lceil M2^{-\frac{R}{M}}-(M-N)^+\rceil
\dot{\geqslant} M(\balpha) =0 $ so there is no outage.

In the second region $\{M(\balpha) \geqslant 1 \}$ the exponent order
of the outage probability depends on the rate. We investigate these
two regions separately.

In the region  $\{\balpha: M(\balpha)=0 \}$, we have $\max_k \alpha_k<1$ since all $\alpha_k 's<1$. From \eqref{eq:PEPBound} and \eqref{eq:M} we conclude that 
\begin{align}
\prob(E,\bar{O}|M(\balpha)=0) &\dot{\leqslant} 2^{Rl} e^{-\rho^{\big(\max_k\alpha_k-1\big)^{-1}}}  \nonumber \\
&=  2^{Rl} e^{-\rho^{\big(1-\max_k\alpha_k\big)}}
\end{align}

Since exponential function dominates all polynomials and $1-\max_k\alpha_k>0$, we get
\begin{equation}
\lim_{\rho \rightarrow \infty}\frac{e^{-\rho^{\big(1-\max_k\alpha_k\big)}}}{\rho^{-MN}}=0 \nonumber
\end{equation}

which in turn yields
\begin{align}
\prob(E,\bar{O}|M(\balpha)=0) &\dot{\leqslant} 2^{Rl} e^{-\rho^{\big(1-\max_k\alpha_k\big)}} \nonumber \\
&\dot{\leqslant} \rho^{-MN} \label{eq:Region1}
\end{align}

We next show that the same result holds for the other region $\{\balpha: M(\balpha) \geqslant 1 \}$.

 Following the same line of argument as we did for \eqref{eq:Region1} but for $M(\balpha) \geqslant 1$, we have 
\begin{align}
P(E,\bar{O}|M(\balpha)\geqslant 1) &\;\dot{\leqslant} \;2^{Rl} e^{-\big( \sum^{L}_{k=1} \frac{\rho\lambda_k}{(1+\rho\lambda_k)^2} \big)^{-1}}\nonumber\\
 &\leqslant e^{2^{Rl}} e^{- \big(\sum_k\frac{1}{1+\rho\lambda} - \sum_k\frac{\rho\lambda_k}{(1:\rho\lambda_k)^2} \big)^{-1} } \nonumber \\
&= e^{2^{Rl}} \underbrace{ e^{- \big(\sum_k\frac{1}{1+\rho\lambda_k} \big)^{-1} } }_{\dot\leqslant 1 \text{ since } M(\balpha) \geqslant 1} \times \nonumber\\
&\qquad e^{\bigg[ -\frac{\sum_k\frac{1}{(\rho\lambda_k+1)^2}}{\big(\sum_k\frac{1}{1+\rho\lambda_k }\big) \big( \sum_k\frac{\rho\lambda_k}{(1+\rho\lambda_k)^2}\big)  } \bigg] }  \nonumber \\
&\;\dot{\leqslant}\; e^{2^{Rl}} e^{\big[- \frac{LM(\balpha)}{LM(\balpha)\rho^{-\min_k|1-\alpha_k|}}\big]} \label{eq:Regiopn2Approx} \\
&\doteq e^{-\rho^{\min_k|1-\alpha_k|}}  \nonumber\\
&\;\dot{\leqslant}\; e^{\rho^{1-\max_k \alpha_k}} \;\dot{\leqslant} \;\rho^{-MN} \label{eq:Region2}
\end{align}
where \eqref{eq:Regiopn2Approx} is direct application of~\eqref{eq:M} for $M(\balpha) \geqslant 1$, and~\eqref{eq:Region2} follows from the fact that $|1-\alpha_k| \geqslant 1$.
Note that \eqref{eq:Region2} is true for any code length $l$.
Invoking the results of \eqref{eq:Region1} and \eqref{eq:Region2}, we can now evaluate $\prob(E,\bar{\mathcal{C}})$ as follows
 \begin{align}
\prob(E,\bar{O}) &=\int\limits_{M(\balpha)=0} \prob(E,\bar{O}|M(\balpha)=0) \prob(\balpha) \text{d}\balpha \nonumber \\
&+ \int\limits_{M(\balpha)\geqslant 1}\prob(E,\bar{O}|M(\balpha) \geqslant 1) \prob(\balpha) \text{d}\balpha \label{nonumber}\\
&\;\dot{\leqslant}\; \rho^{-MN}  \int\limits_{M(\balpha)=0} \prob(\balpha) \text{d}\balpha + \rho^{-MN} \int\limits_{M(\balpha)\geqslant 1}  \prob(\balpha) \text{d}\balpha \\
&\doteq \rho^{-MN}
\end{align}

Therefore,  $\prob(E,\bar{O}) \dot{\leqslant} \rho^{-MN} $ for all regions of $\balpha$.
Finally, \eqref{eq:Perr} becomes
\begin{align}
\prob_{\text{err}}(R,M,N) &\;\dot{\leqslant}\; \prob(E|
O)\;\rho^{-d_{out}} + \prob(E,\bar{O}) \nonumber \\ 
&\leqslant \rho^{-d_{out}} + \prob(E,\bar{O}) \nonumber \\
&\doteq \rho^{-d_{out}} + \rho^{-MN} \nonumber \\
&\doteq  \rho^{-d_{out}} \nonumber \\
&= P_{out}(R,M,N)\label{eq:60}
\end{align}
which establishes the lemma.
\end{Proof}

From Lemma~\ref{lemma:PEPUB} and Lemma~\ref{lemma:PEPLB}, we thus get

\begin{theorem} \label{theorem:MMSE_MIMO}
For MMSE MIMO Receiver under quasi-static channel and joint spatial
encoding, the pairwise error probability (PEP) and the outage
probability $P_{out}$ are exponentially equal and the diversity gain
is $d(R,M,N)=d_{out}(R,M,N)$, where $d_{out}(R,M,N)$ is given in \eqref{eq:d_out}.
\end{theorem}

%------------------------------------------------------------------------------
%--------------------------- MAC --------------------------------
\section{Multiple-Access Channel (MAC)} 
\label{sec:MAC}

We now extend the result to the MAC channel. Consider a MIMO MAC channel with $K$ users, $M$ transmit antennas per
user, $N$ receive antennas (there is no condition on $M, N$ and $k$). Assume flat fading MIMO
channel, the  system model is given by
\begin{equation}
 \by = \sum_{i=1}^{K} \bH_i \bx_i + \bn 
     = \bHeq \bX+ \bn \label{eq:MACmodel}
\end{equation}
where $ {\bH_i \in \mathbb{C}^{N \times M}} $ is the user $i$ channel
matrix whose entries are independent and identically distributed
complex Gaussian, $\bHeq=[\bH_1 \bH_2 \dots \bH_K]$ is the overall
equivalent channel matrix, $ {\bx_i \in \mathbb{C}^{\ M \times 1}}$ is
the transmitted vector of user $i$, $\bX=[\bx_1^T \bx_2^T \dots
  \bx_K^T]^T$ is the overall transmitted vector, and ${\bn \in
  \mathbb{C}^{\ N \times 1}} $ is the Gaussian noise vector. The
vectors ${\bX}$ and ${\bn}$ are assumed independent.  We keep the same
assumptions about the channel. That is we assume a quasi-static flat
fading channel and perfect CSIR and no CSIT. We have the following
theorem

\begin{theorem}
\label{Th:MAC_div}
In a MIMO MAC system with MMSE receiver consisting of $K$ users, $M$
transmit antennas per user and $N$ receive antennas, the lower and
upper bounds on the per user diversity are respectively given by $
d^{\text{MAC}}_{L}(R)$ and $ d^{\text{MAC}}_{U}(R)$,
\begin{align}
 d^{\text{MAC}}_{L}(R)&=\bigg\lceil \big(M 2^{-R/M}-(M-N)^+\big)^+\bigg\rceil^{2}+\nonumber\\
&\hspace{20pt} \big|N-KM\big|\bigg\lceil \big(M
 2^{-R/M}-(M-N)^+\big)^+\bigg\rceil \label{eq:MAC_DRT_Th}\\ 
d^{\text{MAC}}_{U}(R)&=\bigg\lceil \big( KM 2^{-R/{KM}}-(M-N)^+\big)^+\bigg\rceil^{2}+\nonumber\\
&\hspace{20pt} \big|N-KM\big|\bigg\lceil \big(KM
 2^{-R/{KM}}-(M-N)^+\big)^+\bigg\rceil. \label{eq:MAC_DRT_Th2}
\end{align}
\end{theorem}

From~\eqref{eq:MAC_DRT_Th} it is straightforward to verify the single
user case. The machinery of the proof is mostly similar to the single
user case. However, the outage upper and lower bounds are obtained in
a different manner that is pointed out in the following analysis for
$N\geqslant M$. The case $ N<M$ can be similarly obtained.

%-----------------------------------------------------------------------------
\subsection{MAC Outage Upper Bound}
\label{subsec:MAC_DRT_UB}

The user $i$ outage probability can be written as
\begin{equation}
P^i_{out} = \prob\bigg(\sum^{iM}_{k=(i-1)M+1}\log(1+\gamma^i_k)<R\bigg). 
%&= \prob\bigg(\prod^M_{k=1}(1+\gamma_k)<2^R\bigg) \label{eq:11
\label{eq:MAC_MMSEoutage}
\end{equation}
where $\gamma^i_k$ is the SINR of the stream $k$ of user $i$.
Specializing this to MMSE receiver we get
\begin{equation}
P^i_{out} = \prob\bigg(\sum^{iM}_{k=(i-1)M+1}\log(\bI+\rho \bHeq^{H}\bHeq)^{-1}_{kk}>-R\bigg).
\label{eq:MACUpperBoundSum}
\end{equation}

Using Jensen's Inequality the outage probability can be bounded as
\begin{align}
P^{i}_{out}&\leqslant \prob\big(
\log\big(\!\!\!\!\!\!\! \sum^{iM}_{k=(i-1)M+1} \!\! \frac{1}{M}(\bI+\rho \bHeq^H\bHeq)^{-1}_{kk}
\big)>\frac{-R}{M}\big) \nonumber \\ 
&\leqslant \prob\big(
\log\big(\sum^{KM}_{k=1}\frac{1}{M} (\bI+\rho \bHeq^H\bHeq)^{-1}_{kk}
\big)>\frac{-R}{M}\big)  \label{eq:MACDRTUpper}  \\
&= \prob\big( \sum^{KM}_{k=1} \frac{1}{1+\rho \lambda_{k}} > M2^{-\frac{R}{M}}\big) \label{eq:MACDRTEigen} 
\end{align}
where~\eqref{eq:MACDRTUpper} is true since the summation in the
left-hand side of the inequality adds more positive terms (recall that
$(\bI+\rho \bHeq^H\bHeq)$ is a positive definite
matrix~\cite{Telatar:EUR99}). 
Following similar steps that were used to obtain~\eqref{eq:DivRes} we can easily show that $P^i_{out} \;\dot{\leqslant}\; \rho^{-d^{\text{MAC}}_L}$, where $d^{\text{MAC}}_L$ is given by~\eqref{eq:MAC_DRT_Th}.

%-----------------------------------------------------------------------------
\subsection{MAC Outage Lower Bound}
\label{subsec:MAC_DRT_LB}

The outage probability can be lower bounded as follows
\begin{align}
P^i_{out} &= \prob\bigg(\sum^{iM}_{k=(i-1)M+1}\log(\bI+\rho
\bHeq^{H}\bHeq)^{-1}_{kk}>-R\bigg) \nonumber \\ &\geqslant
\prob\bigg(\sum^{KM}_{k=1} \log(\bI+\rho \bHeq^H \bHeq) ^{-1}_{kk}
>-R\bigg) \label{eq:MACDMTLBTrivial} \\ 
 &\;\dot{\geqslant}\; \prob\bigg(\sum^{KM}_{k=1}(\bI+\rho \bHeq^H\bHeq)^{-1}_{kk}>\frac{KM}{a}2^{\frac{-R}{M}} \bigg)\label{eq:JensLowBou}
\end{align}
where \eqref{eq:MACDMTLBTrivial} is a trivial bound based on
dedicating all $KM$ antennas to {\it one} user, and
\eqref{eq:JensLowBou} uses the same technique as in
Section~\ref{sec:LB}, and $a$ is a positive number slightly less than
one.
%% the Specht inequality with Specht ratio given
%% by~\eqref{eq:SpechtRatio}.
Following similar steps that were used to
obtain~\eqref{eq:DivRes} we can easily show that $P^i_{out}
\;\dot{\geqslant}\; \rho^{-d^{\text{MAC}}_U}$, where
$d^{\text{MAC}}_U$ is given by~\eqref{eq:MAC_DRT_Th2}.

%------------------------------------------------------------------------------
\section{Frequency-Selective Channel} 
\label{sec:freqSel}

Broadband wireless systems usually operate in frequency-selective
channels where, in addition to the spatial diversity obtained in MIMO
broadband systems, frequency diversity can be achieved. Broadband
systems usually employ orthogonal frequency division multiplexing
(OFDM) or single carrier (SC)
transmission~\cite{AlDhahir:CommLet01}. Specifically, SC was shown to
be attractive for broadband wireless channels due to its lower
complexity, lower peak-to-average power ratio and reduced sensitivity
to carrier frequency errors compared to
OFDM~\cite{Sari:CommMag95,AlDhahir:CommLet01}.

In this section, we investigate the diversity achieved by SC-MMSE
receivers for two block transmission schemes, namely cyclic prefix
(CP) and zero-padding (ZP) schemes. The CP and ZP are commonly used
for guard intervals in block quasi-static channels. Although CP was
initially proposed for both single carrier and multi-carrier systems,
ZP was lately shown to be an attractive alternative for both
systems~\cite{Muquet:JC02,Zhendao:WCNC02}.

%------------------------------------------------------------------------------
\subsection{System Model} 
\label{sec:freqSelModel}

We consider a general MIMO system in a rich scattering quasi-static
environment. The equivalent baseband channel is given by multipath
model with $\nu$ paths referred to as the ISI channel in the
sequel. The $(\nu+1)$-tap channel impulse response between the
transmit antenna $m$ and receive antenna $n$ is denoted by the vector
$\bh_{mn}=[h_{mn,0},h_{mn,1},\dots, h_{mn,\nu}]$. We assume a
block-fading model where $\bh_{mn}$ remains unchanged during a
transmission block.
Assuming $M$ transmit and $N$ receive antennas, the received vector $\by_{k}$ at time instant $k$ is given by~\cite{Scaglione:JSP02,Tajer:WCOM10}
\begin{equation}
\by_{k}=\sum^{\nu}_{i=0}\bH_i\bx_{k-i}+\bn_k\label{eq:FreqSelModel}
\end{equation}
where $\bH_i$ is the $M\times N$ channel matrix that has $h_{mn,i}$ as
its $(m,n)$ element, $\bx_{k-i}$ is $M\times 1$ transmitted vector at
time index $k-i$, $\by_k$ is the $N\times 1$ received vector and
$\bn_k$ is the $N\times 1$ Gaussian noise vector at time index $k$.

Consider a transmission of $L_d+L_e$ spatial vectors each of size
$M\times 1$, where $L_d$ is an integer representing the number of
transmissions over the quasi-static channel and $L_e$ is the length of
data extension to avoid inter-block interference, in the form of
either zero-padding or cyclic prefix. The receiver discards the $L_e$
vectors in the case of cyclic-prefix
transmission~\cite{Scaglione:JSP02}. Stacking the transmitted vector
in an $M(L_d+L_e) \times 1$ vector, we can write the stacked
$M(L_d+L_e)\times 1$ transmitted as follows
\begin{align}
\bar{\bx}_k&=[\bx^T_{k(L_d+L_e)},\dots, \bx^T_{k(L_d+L_e)+L_d+L_e-1} ]\nonumber%% \\
%% \bar{\by}_k&=[\by^T_{k(L_d+L_e)+L_e},\dots, \by^T_{k(L_d+L_e)+L_d+L_e-1} ].\nonumber
\end{align}

We can then rewrite~\eqref{eq:FreqSelModel} as
\begin{equation}
\bar{\by}_{cp}=\bar{\bH}\,\bar{\bx}+\bar{\bn} \label{eq:FreqSelModelMod}
\end{equation}
where $\bar{\by}_{cp}$ is the $NL_d\times 1$ received vector, $\bar{\bx}$ is the $M(L_d+L_e)\times 1$ transmitted vector, $\bar{\bn}$ is the white Gaussian noise vector $\in \mathbb{C}^{\ NL_d \times 1}$ and $\bar{\bH}$ is the channel matrix given by
\begin{equation}
\bar{\bH}=
\begin{bmatrix} \bH_0 &\bH_1 &\cdots &\bH_{\nu} &0 &\cdots &0 \\
                0  &\bH_0 &\bH_1 &\cdots &\bH_{\nu} &\cdots &0 \\
                \vdots  &\ddots &\ddots &\ddots &\ddots &\vdots \\
                0  &\cdots &\cdots &\bH_0  &\bH_1  &\cdots &\bH_{\nu}
\end{bmatrix}.\label{eq:ChMat}
\end{equation}

The linear data extension operation maps the data vector $\hat{\bx}$ to the transmitted vector $\bar{\bx}$ and is shown by
\begin{equation}
 \bar\bx = \mathbf{U_{cp}}\hat{\bx}%% \quad \bar\bx =
%%  \mathbf{U_{zp}} \,\text{ Flip}(\hat{\bx})
\end{equation}
where $\mathbf{U_{cp}}$ is given by
\begin{equation}
 \mathbf{U_{cp}}=
 \left[
  \begin{array}{ c c }
     \bI_{ML_d} \quad\quad\quad \\
     \bI_{ML_e} \quad\quad \b0_{ML_e \times (L_d-L_e)M}
  \end{array} \right] \label{eq:CP_matrix}
\end{equation}
The system model in~\eqref{eq:FreqSelModelMod} can now be written in
terms of the unpadded data vector $\hat{\bx}$ and an equivalent
channel matrix $\bHeq$ as follows
\begin{equation}
\bar\by_{cp}={\bHeq}\,\hat{\bx}+\bar\bn \label{eq:Modifmodel}
\end{equation}
where in a CP system, $\bHeq=\bar\bH\mathbf{U_{cp}}$ is a $NL_d\times
ML_d$ block circulant matrix constructed by block circulations of the
matrix $[\bH_0,\bH_1,\dots, \bH_{\nu},0,\dots ,0]^T$. 

For the zero-padding transmission, we can
rewrite~\eqref{eq:FreqSelModel} as
\begin{equation}
\bar{\by}_{zp}=\bHeq\,\hat{\bx}+\bar{\bn} \label{eq:FreqSelModelModZP}
\end{equation}
where $\bar{\by}_{zp}$ is the $N(L_d+L_e)\times 1$ received vector,
$\bar{\bx}$ is the $ML_d\times 1$ transmitted vector, $\bar{\bn}$ is
the white Gaussian noise vector $\in \mathbb{C}^{ N(L_d+L_e) \times
  1}$ and $\bar{\bH}$ is the channel matrix given by
\begin{equation}
\bHeq=
\begin{bmatrix} \bH_{0}   &0        &\cdots  &0  \\
                \vdots    &\bH_{1}  &\ddots &\vdots \\
                \bH_{\nu}      &\vdots   &\ddots &\bH_{0}  \\
                0         &\bH_{\nu}    &\ddots &\vdots  \\
                \vdots      &\vdots  &\vdots &\bH_{\nu} 
\end{bmatrix}.\label{eq:ChMatZp}
\end{equation}

Assuming perfect channel state information at the receiver (CSIR) and
that the channel remains unchanged during the transmission of
$L_d+L_e$ vectors, the MMSE equalizer $\bW$ is applied to decouple the
received streams (after removing the $L_e$ extension vectors in
case of cyclic-prefix transmission). The MMSE equalizer is given by
\begin{equation}
 \bW = (\rho^{-1} \bI+\bHeq^{H}\bHeq)^{-1}\bHeq^{H} \label{eq:EqMod}
\end{equation}
and the unbiased decision-point SINRs of the equalizers output  for detecting the $k^{th}$ transmitted stream are
\begin{equation}
 \gamma_k =\frac{1}{(\bI+\rho \bHeq^H\bHeq)^{-1}_{kk}}-1 \quad \quad k=1,\dots,ML_d \label{eq:SINR}.
\end{equation}

In the following sections we analyze the outage diversity for the ZP
and CP systems. The PEP analysis follows in a direct manner as in the flat fading case so we omit it.
%------------------------------------------------------------------------------
\subsection{The Zero Padding MMSE Receiver} 
\label{sec:ZP}

It is known that in a point-to-point {\em single-antenna} ISI channel,
linear receivers can achieve full multipath diversity under
zero-padding
transmission~\cite{Tepedel:ICASSP03,Grokop:JIT09,Zhendao:WCNC02}. In
this section we investigate the similar question for MIMO systems
whose receivers use linear MMSE operations in both the spatial and
temporal dimensions. We provide lower and upper bounds on
diversity. The bounds are not always tight, but the diversity is fully
characterized for SIMO systems.

We begin by analyzing the tradeoff between the spectral
efficiency $R$ and the diversity of MMSE receiver in the {\em single-antenna}
ISI channel $d^{ISI}_{MMSE}$ under ZP transmission. Tajer {\it et
  al}~\cite{Tajer:WCOM10} shows that $d^{ISI}_{MMSE}$ varies with $R$
under CP transmission and MMSE equalization, in particular, for a
quasi-static single-antenna ISI channel with $\nu +1$ taps, the
diversity of the SC-MMSE receiver under CP transmission is
$d^{CP}_{MMSE}=1+\min(\nu,\lfloor 2^{-R}L_d\rfloor )$, where $L_d$ is the
transmission data block length. We show that the same is not true for
ZP transmission.

\begin{lemma} \label{lemma:Div_CP_ZP_ISI}
For a quasi-static single-antenna ISI channel with $\nu +1$ taps, the
diversity of the SC-MMSE receiver under ZP transmission is
$d^{ZP}_{MMSE}=\nu+1$ irrespective of $R$.
\end{lemma}

\begin{proof}
%% \subsection{Proof of Lemma~\ref{lemma:Div_CP_ZP_ISI}}\label{Div_ISI_ZP}
See Appendix~\ref{appendix:CP_ZP_ISI_APP}.
\end{proof}

%----------------------------------------------------------------------------

We proceed with lower and upper bounds on diversity for MIMO ISI channel.

\subsubsection{Diversity Upper Bound}
\label{sec:ZPLB}
Applying the MMSE equalizer given by~\eqref{eq:EqMod} to the received
vector in~\eqref{eq:Modifmodel}, the effective mutual information
between $\hat{\bx}$ and $\bW\bar{\by}$ is equal to the sum of mutual
information of their components~\cite{Onggosanusi:ICAS02}
\begin{equation}
 I(\hat{\bx},\bW\bar{\by})=\frac{1}{L_d}\sum^{ML_d}_{k=1} I(x_k,y_k).
\nonumber
\end{equation}
 Thus the outage probability is given by
\begin{align}
P_{out} &= \prob\bigg(\frac{1}{L_d}\sum^{ML_d}_{k=1}\log(1+\gamma_{k})<R\bigg)  
\label{eq:ZP_SINR}\\
&= \prob\bigg(\frac{1}{L_d}\sum^{ML_d}_{k=1}\log\frac{1}{(\bI+\rho \bHeq^{H}\bHeq)^{-1}_{kk}}<R\bigg)
\label{eq:ZPMMSE2}\\
&\geqslant \prob\bigg(M\log \frac{1}{ML_d}\sum^{ML_d}_{k=1}\frac{1}{(\bI+\rho \bHeq^{H}\bHeq)^{-1}_{kk}}<R\bigg)
\label{eq:ZPMMSE3}
%% \\
%% &\geqslant \prob\big(  M \log \frac{1}{M L_dM_{\beta}}\sum^{ML_d}_{k=1}(\bI+\rho \bHeq^{H}\bHeq)^{-1}_{kk}>-R  \big)\label{eq:ZP_Specht} \\ 
%% &\dot{=} \hspace{4pt} \prob\big(M \log \frac{1}{ML_d}\sum^{ML_d}_{k=1}(\bI+\rho \bHeq^{H}\bHeq)^{-1}_{kk}>-R  \big) \label{eq:ZP_MMSE_Bound}\\
%% &= \hspace{4pt} \prob\big(M \log \frac{1}{ML_d}\sum^{ML_d}_{k=1}\frac{1}{1+\rho\lambda_k} >-R  \big) \label{eq:ZP_MMSE_eigen}
\end{align}
where we have used Jensen's inequality as in Section~\ref{sec:LB}.  Let
the eigen decomposition of $\bHeq^H\bHeq$ be given by
$\bHeq^H\bHeq=\bU^H\Lambda\bU$ where $\bU$ is unitary and $\Lambda$ is
a diagonal matrix that has the eigenvalues of the matrix
$\bHeq^H\bHeq$ on its diagonal. Let the eigenvalues of $\bHeq^H\bHeq$
be given by $\{\lambda_{\ell}\}$ with $\lambda_1\geqslant
\lambda_2\cdots\geqslant \lambda_{ML_d}$. Let the vector $\bu_k$ be
the column $k$ of the matrix $\bU$, we have
\begin{align}
(\bI+\rho \bHeq^H\bHeq)^{-1}_{kk} &= \bu^H_k (\bI+\rho \Lambda)^{-1} \bu_k\nonumber\\
&= \sum_{\ell=1}^{ML_d} \frac{|u_{\ell k}|^2}{1+\rho\lambda_{\ell}}\nonumber\\ &\triangleq  S_k. \nonumber
\end{align}
Let $\bar{k}=\arg\min_k S_k$.  we can bound the sum
in~\eqref{eq:ZPMMSE3}
\begin{align}
\frac{1}{ML_d} \sum_{k=1}^{ML_d} \frac{1}{(\bI+\rho \bHeq^H\bHeq)^{-1}_{kk}}  &= \frac{1}{ML_d} \sum_{k=1}^{ML_d}\frac{1}{S_k} \nonumber\\
  &\leqslant \;  \frac{1}{\min_k S_k} \nonumber \\
&= \; \frac{1}{ S_{\bar{k}}}
\end{align}
thus the outage bound in~\eqref{eq:ZPMMSE3} can be further bounded
\begin{align}
P_{out}  &\geqslant \prob\bigg( M \log \frac{1}{ML_d}\sum^M_{k=1}\frac{1}{(\bI+\rho \bHeq^H\bHeq)^{-1}_{kk}}< R\bigg) \nonumber\\
&\geqslant \prob\bigg( M \log \frac{1}{ S_{\bar{k}}}< R\bigg) \nonumber\\
&=  \prob\bigg( S_{\bar{k}} >  2^{-\frac{R}{M}} \bigg)  \label{eq:OutJenLowerZP}
\end{align}

We now bound~\eqref{eq:OutJenLowerZP} by conditioning on the event
\begin{align}
\mathcal{B}\triangleq &\bigg\{ %% |u_{\ell \bar{k}}|^2 \geqslant
%% \epsilon_2, \; \ell=1,\cdots ML_d-M \nonumber\\ &\text{ and } 
|u_{\ell \bar{k}}|^2 \geqslant
\frac{a}{M},\; \ell=ML_d-M+1,\cdots,ML_d \bigg\}
\end{align}

where $a$ is a positive real number that is slightly smaller than one
$a=1-\epsilon_1$, and $\epsilon_1$ is a small positive number. We then
have
\begin{align}
P_{out}&=\prob\bigg( S_{\bar{k}} > 2^{-\frac{R}{M}} \bigg) \nonumber\\
 &\geqslant
\prob\bigg( S_{\bar{k}} > 2^{-\frac{R}{M}} \big|\mathcal{B} \bigg)
\prob(\mathcal{B})\nonumber\\
 &=\prob\bigg( \sum_{\ell=1}^{ML_d} \frac{|u_{\ell \bar{k}}|^2}{1+\rho\lambda_{\ell}} > 2^{-\frac{R}{M}} \bigg|\mathcal{B} \bigg)
\prob(\mathcal{B}) \nonumber\\
 &\geqslant \prob\bigg( \sum_{\ell=ML_d-M+1}^{ML_d} \frac{|u_{\ell \bar{k}}|^2}{1+\rho\lambda_{\ell}} > 2^{-\frac{R}{M}} \bigg|\mathcal{B} \bigg)
\prob(\mathcal{B}) \label{eq:Ellim}\\
 &\geqslant\prob\bigg( \frac{1}{M}\sum_{\ell=ML_d-M+1}^{ML_d} \frac{a}{1+\rho\lambda_{\ell}} > 2^{-\frac{R}{M}}  \bigg) \prob(\mathcal{B}) \nonumber\\
&\doteq \prob\bigg( \frac{1}{M}\sum_{\ell=ML_d-M+1}^{ML_d} \frac{a}{1+\rho\lambda_{\ell}} > 2^{-\frac{R}{M}}  \bigg) \label{eq:AsymptSim}\\
&= \prob\bigg(\sum_{\ell=ML_d-M+1}^{ML_d} \frac{1}{1+\rho\lambda_{\ell}} > \frac{M}{a}2^{-\frac{R}{M}}  \bigg) \label{eq:OutLOwBou}
\end{align}
where~\eqref{eq:Ellim} follows by removing some of the elements of the
sum corresponding to the largest eigenvalues. The steps used to obtain
Eq.~\eqref{eq:AsymptSim} are similar to the steps used in
Section~\ref{sec:LB}. 

 Note that $\bHeq^{H}\bHeq$ is not a Wishart matrix, hence the
 analysis of Section~\ref{sec:LinearRx} does not directly apply
 here. The block diagonal elements of $\bHeq^{H}\bHeq$ are similar and
 are given by
\begin{equation}
\mathbf{D}=\sum^{\nu}_{i=0}\bH_i^{H}\bH_i.\label{eq:BlockDiag}
\end{equation}

The matrix $\bHeq^{H}\bHeq$ is Toeplitz and Hermitian. Moreover, the
matrix $\mathbf{D}$ given by~\eqref{eq:BlockDiag} is a Wishart
matrix\footnote{\hspace{2pt} Let $\mathcal{W}(n,\mathcal{\sum})$
  denote a Wishart distribution with degree of freedom $n$ and
  covariance (also called scale) matrix $\mathcal{\sum}$. Any of the
  diagonal block matrices $\mathbf{D}_j$ given by~\eqref{eq:BlockDiag}
  follows a Wishart distribution since if $\mathbf{B}_1\in
  \mathcal{W}(n_1,\mathcal{\sum})$ and $\mathbf{B}_2\in
  \mathcal{W}(n_2,\mathcal{\sum})$ then
  $\mathbf{B}_1+\mathbf{B}_2\in\mathcal{W}(n_1+n_2,\mathcal{\sum})$.}.

Observe that the probability in~\eqref{eq:OutLOwBou} depends on the
$M$ smallest eigenvalues. We now bound these eigenvalues with the
eigenvalues of the matrix $\bD$ via the Sturmian separation
theorem~\cite[P.1077]{Gradshtey:book}.
\begin{theorem}\label{Th:Sturmian}
{\em (Sturmian Separation Theorem)} Let $\{\mathbf{A}_r,
r=1,2,\ldots\}$ be a sequence of symmetric $r\times r$ matrices such
that each $\mathbf{A}_r$ is a submatrix of $\mathbf{A}_{r+1}$. Then if
$\{\lambda_k{(\mathbf{A}_r)} \; , \; k=1,\ldots, r \}$ denote the
ordered eigenvalues of each matrix $\mathbf{A}_r$ in descending order, we have
\begin{equation}
\lambda_{k+1}(\mathbf{A}_{i+1}) \leq \lambda_k(\mathbf{A}_i)  \leq \lambda_k(\mathbf{A}_{i+1}).\nonumber
\end{equation}
\end{theorem}
\vspace{0.1in}

For our purposes, we consider a special case of the Sturmian Theorem
by constructing a set of matrices $\mathbf{A}_M, \mathbf{A}_{M+1},
\ldots, \mathbf{A}_{L_dM}$ starting by the largest one
$\mathbf{A}_{L_dM} \defeq \bHeq^H\bHeq$ and making all other matrices
$\mathbf{A}_i$ to be (successively embedded) $i\times i$ principal
submatrices of $\bHeq^H\bHeq$, such that the smallest matrix is $\mathbf{A}_M =
\mathbf{D}_{L_d}$. Then we repeatedly apply the first inequality in the Sturmian to get:
\begin{alignat*}{2}
\lambda_{ML_d}(\mathbf{A}_{ML_d}) &\leq \lambda_{ML_d-1}(\mathbf{A}_{ML_d-1}) &\leq \cdots &\leq \lambda_M(\mathbf{A}_{M})\nonumber\\
\lambda_{ML_d-1}(\mathbf{A}_{ML_d}) &\leq \lambda_{ML_d-2}(\mathbf{A}_{ML_d-1}) &\leq \cdots &\leq \lambda_{M-1}(\mathbf{A}_{M})\nonumber\\
&\vdots &\vdots  \nonumber\\
\lambda_{ML_d-M+1}(\mathbf{A}_{ML_d}) &\leq \lambda_{ML_d-M}(\mathbf{A}_{ML_d-1})&\leq \cdots &\leq \lambda_{1}(\mathbf{A}_{M})\label{eq:SturmianLB}
\end{alignat*}
This implies that the smallest $M$ eigenvalues of $\bHeq^H\bHeq$ are
bounded above by the $M$ eigenvalues of $\mathbf{D}$, respectively. Hence:
\begin{equation}
 P_{out}\dot{\geqslant}\hspace{4pt} \prob\big(\sum^{M}_{k=1}\frac{1}{1+\rho\lambda_k(\mathbf{D})} >\frac{M}{a}2^{-\frac{R}{M}}
 \big) \label{eq:ZP_MMSE_eigen_bound2}.
\end{equation}
$\mathbf{D}$ is a sum of $(\nu+1)$ central Wishart matrices each
with $N$ degrees of freedom and with identity covariance matrix,
i.e. $\mathbf{D}\in \mathcal{W}((\nu+1) N,I)$. Therefore the
analysis of Section~\ref{sec:LinearRx} applies here and we have the
following lemma.

\begin{lemma}
\label{Lemma:ZPLB}
In a MIMO quasi-static frequency-selective system (with channel memory
$\nu$) consisting of $M$ transmit and $N$ receive antennas, the MMSE
receiver diversity under joint spatial encoding and zero-padding transmission
is upper bounded as
\begin{align}
d^{ZP}\leqslant&\bigg\lfloor
 \big( M2^{-\frac{R}{M}}+1-(M-N)^+\big)^+ \bigg\rfloor^2\nonumber\\
&+\big|(\nu+1)N-M \big| \bigg\lfloor
 \big(M2^{-\frac{R}{M}}+1-(M-N)^+\big)^+ \bigg\rfloor \label{eq:DivZpUB2}
\end{align}

\end{lemma}

\subsubsection{Diversity Lower Bound}
\label{sec:ZPUB}

We can upper bound the outage probability as follows.
\begin{align}
&P_{out} =
\prob\bigg(\frac{1}{L_d}\sum^{ML_d}_{k=1}\log(1+\gamma_{k})<R\bigg)
\nonumber\\ &= \prob\bigg(\frac{1}{L_d}\sum^{ML_d}_{k=1}\log(\bI+\rho
\bHeq^{H}\bHeq)^{-1}_{kk}>-R\bigg) \nonumber\\ &\leqslant \prob\bigg(
M \log \frac{1}{ML_d}\sum^{ML_d}_{k=1}(\bI+\rho
\bHeq^{H}\bHeq)^{-1}_{kk}>-R \bigg)\label{eq:ZP_Jens}
\\ &\leqslant \prob\bigg(
M \log \frac{1}{M}\sum^{ML_d}_{k=1}(\bI+\rho
\bHeq^{H}\bHeq)^{-1}_{kk}>-R \bigg)\nonumber\\
&= \hspace{4pt} \prob\bigg(
\sum^{ML_d}_{k=1}\frac{1}{1+\rho\lambda_k(\bHeq^{H}\bHeq)}
>M2^{-\frac{R}{M}} \bigg) \nonumber\\ &\leqslant \hspace{4pt}
\prob\bigg(\sum^{M}_{k=1}\frac{1}{1+\rho\lambda_k(\bHeq^{H}\bHeq)}+L_dM-M
>M2^{-\frac{R}{M}} \bigg) \label{eq:ZP_MMSE_eigen2}\\ &= \hspace{4pt}
\prob\bigg(\sum^{M}_{k=1}\frac{1}{1+\rho\lambda_k(\bHeq^{H}\bHeq)}
>M2^{-\frac{R}{M}}-(ML_d-M) \bigg) \label{eq:ZP_MMSE_eigen3}
\end{align}
where~\eqref{eq:ZP_Jens} follows from Jensen's inequality and \eqref{eq:ZP_MMSE_eigen2} follows from setting the smallest $L_dM-M$ eigenvalues to zero.

Now we repeatedly use the second inequality in the Sturmian theorem to get
\begin{alignat*}{2}
\lambda_{M}(\mathbf{A}_{M}) &\leq \cdots
&\leq\lambda_{M}(\mathbf{A}_{ML_d-1}) &\leq
\lambda_M(\mathbf{A}_{ML_d})\nonumber\\ 
\lambda_{M-1}(\mathbf{A}_{M})
&\leq \cdots &\leq \lambda_{M-1}(\mathbf{A}_{ML_d-1}) &\leq
\lambda_{M-1}(\mathbf{A}_{ML_d})\nonumber\\ 
&\vdots &\vdots
\nonumber\\ \lambda_{1}(\mathbf{A}_{M}) &\leq \cdots &\leq
\lambda_{1}(\mathbf{A}_{ML_d-1})&\leq
\lambda_{1}(\mathbf{A}_{ML_d})\label{eq:SturmianLB}
\end{alignat*} 
with $\mathbf{A}_{ML_d} \defeq\bHeq^{H}\bHeq$ and
$\mathbf{A}_{M}\defeq\mathbf{D}$, similar to the earlier
case. Therefore the largest $M$ eigenvalues of $\bHeq^{H}\bHeq$ are
bounded below by the $M$ eigenvalues of $\mathbf{D}$,
respectively. Therefore
\begin{equation}
 P_{out}\dot{\leqslant}\hspace{4pt} \prob\bigg(M \log
 \frac{1}{M}\sum^{M}_{k=1}\frac{1}{1+\rho\lambda_k(\mathbf{D})} > Q
 \bigg) \label{eq:ZP_MMSE_eigen_LBbound2}.
\end{equation}
where $Q=\max\big(0,M2^{-\frac{R}{M}}-(ML_d-M)\big)$.
Recall that $\mathbf{D}$ is a Wishart matrix, therefore the
analysis of Section~\ref{sec:LinearRx} follows and we obtain the
following lemma.
 
\begin{lemma}
\label{Lemma:ZPUB}
In a MIMO quasi-static frequency-selective system (with channel memory
$\nu$) consisting of $M$ transmit and $N$ receive antennas, the MMSE
receiver diversity is lower bounded as
\begin{equation}
d^{ZP}\geqslant\big\lceil
 Q\big\rceil^2+|(\nu+1)N-M|\big\lceil
 Q\big\rceil \label{eq:DivZpLB2}
\end{equation}
 under joint spatial encoding and zero-padding transmission.
 $Q=\max\big(0,M2^{-\frac{R}{M}}-(ML_d-M)\big)$.
\end{lemma}
  \begin{remark}
Notice that both lower and upper bounds differ only in the second term
of $Q$, i.e. ($ML_D-M$). The diversity lower bound for $L_d=1$ is
tight against the upper bound, but for $L_d>1$ the lower
bound \eqref{eq:DivZpLB2} is trivial.
\end{remark}

%------------------------------------------------------------------------------
\subsection{The Cyclic Prefix MMSE Receiver} 
\label{sec:CP}

\begin{figure*}
\centering
\includegraphics[width=6.5in]{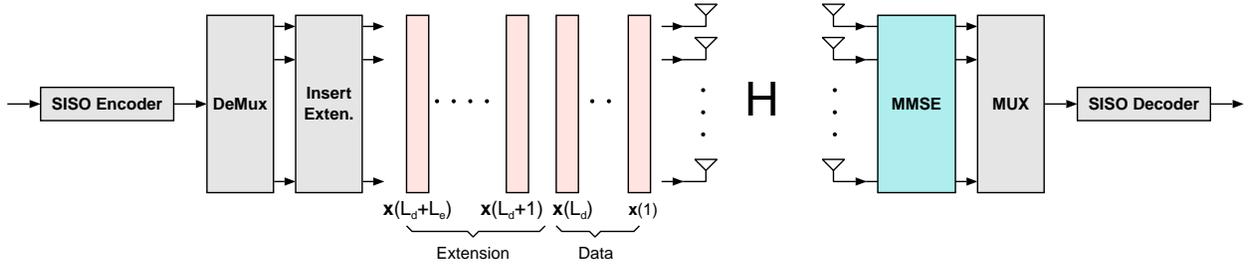}
\caption{Single-carrier block transmission in a frequency-selective channel. In the case of CP, the extension is removed at the receiver prior to equalization.}
\label{fig:freqsel}
\end{figure*}

For the {\em single-antenna} ISI channel under CP transmission, the
explicit tradeoff between spectral efficiency and diversity was
found~\cite{Tajer:WCOM10} to be $d^{CP}_{MMSE}=1+\min(\nu,\lfloor
2^{-R}L_d\rfloor)$. In this section, we extend the analysis to the MIMO
case.  The system model is shown in Figure~\ref{fig:freqsel}. We start
with the general $M\times N$ MIMO system.

%% \subsubsection{The MIMO Channel under CP Transmission}
%% \label{sec:GeneralCP}

The system model is again given by~\eqref{eq:Modifmodel} where
$\bHeq=\bar\bH\mathbf{U_{cp}}$  and $\hat{\bx}$ is
  generated by taking the IDFT of the information vector
  $\bx$~\cite{Stamoulis:TSP02}, i.e.
\begin{equation}
\hat{\bx}=\bQ_{Tx}^H\bx
\end{equation}
where $\bQ_{Tx}$ is the augmented DFT matrix given by
$\bQ_{Tx}=\bQ\otimes\bI_M$, where
$\bI_M$ is the identity matrix,
$\bQ$  is the normalized DFT matrix, and $\otimes$ is the Kroenecker product.

The $NL_d\times ML_d$ block-circulant matrix $\bHeq$ has eigen
decomposition $\bHeq=\bQ_{Rx}^H{\bf{\Lambda}}\bQ_{Tx}$, where $\bQ_{Rx}=\bQ\otimes\bI_N$. Both $\bQ_{Tx}$ and $\bQ_{Rx}$ are unitary matrices.
The block diagonal matrix ${\bf \Lambda}$ is given by

\begin{equation}
{\bf \Lambda}=\begin{pmatrix}
{\bf B}_1&{}&{}&{0}\\ 
{}&{\bf B}_2&{}&{}\\ 
{}&{}&{\ddots}&{}\\ 
{0}&{}&{}&{\bf B}_{L_d}
\end{pmatrix}
\end{equation}
where the matrix ${\bf B}_k$ is given by~\cite{Kaveh:Spring10}
\begin{equation}
{\bf B}_k=\sum_{i=0}^{\nu}\bH_i e^{-j\frac{2\pi i(k-1)}{L_d}}\quad {\text{for }} k=1,\dots,L_d \label{eq:LambdaEq}
\end{equation}
and $\bH_i$ is the instantaneous MIMO channel (cf. Section~\ref{sec:freqSelModel}).

Analogous to the proof of~\cite{Tajer:WCOM10}, we first consider the
case where the transmission data-block length is equal to the number of
channel taps, i.e. $L_d=\nu+1$. In this case, the entries
of ${\mathbf B}_k{'}s$ are i.i.d. normal complex Gaussian. 
%We start with the outage upper bound. 

\subsubsection{Outage upper bound}

The outage probability of the MMSE receiver  is given by
\begin{align}
P_{out}&=\prob \bigg( \frac{1}{L_d} \sum_{k=1}^{ML_d} \log
(\frac{1}{(\bI+\rho\bHeq^H\bHeq)^{-1}_{kk}}) < R \bigg)
\label{eq:CPMMSEOuEq}\\&=\prob \bigg( \frac{1}{L_d} \sum_{k=1}^{ML_d} \log
((\bI+\rho\bHeq^H\bHeq)^{-1}_{kk}) > - R \bigg)
\nonumber\\
 &\leqslant \prob \bigg( M\log \sum_{k=1}^{ML_d}
\frac{1}{ML_d}(\bI+\rho\bHeq^H\bHeq)^{-1}_{kk} > -R
\bigg) \label{eq:CPJen}\\ &=  \prob \bigg( M\log
\sum_{k=1}^{ML_d} \frac{1}{ML_d}(\bI+\rho{\bf \Lambda}^H{\bf \Lambda})^{-1}_{kk} > -R
\bigg) \label{eq:CPDiag}\\ &=  \prob \bigg(
\sum_{k=1}^{ML_d} (\bI+\rho{\bf \Lambda}^H{\bf \Lambda})^{-1}_{kk} > ML_d2^{-\frac{R}{M}}
\bigg) \nonumber\\
&=  \prob \bigg(
\sum_{i=1}^{L_d} \text{tr}(\bI+\rho {\bf B}_i^H{\bf B}_i)^{-1} > ML_d2^{-\frac{R}{M}}
\bigg) \nonumber\\
&=  \prob \bigg(
\sum_{i=1}^{L_d} \sum_{k=1}^{M}\frac{1}{(1+\rho \lambda_{k,i})} > ML_d2^{-\frac{R}{M}}
\bigg) \label{eq:CPEig}
\end{align}
Where~\eqref{eq:CPJen} follows from Jensen's inequality,
\eqref{eq:CPDiag} follows from the eigen decomposition of $\bHeq$, and
$\lambda_{k,i}$ is $k$-th eigenvalue of the $i$-th Wishart matrix
${\bf B}_i^H{\bf B}_i$.

Recall from Section~\ref{sec:OutageAnalysis} that the eigenvalues of
a Wishart matrix have the asymptotic property
 \begin{equation}
\sum^M_{k=1}\frac{1}{1+\rho \lambda_k} \doteq
\sum_{\alpha_k>1}1+\sum_{\alpha_k<1}\rho^{\alpha_k-1} \label{eq:CP_MIMO}
\end{equation}
based on which we established in Lemmas~\ref{lemma:outageUB} and \ref{lemma:outageLB} the following
\begin{equation}
\prob \big( \sum^M_{k=1}\frac{1}{1+\rho \lambda_k}\geqslant s \big) \doteq \rho^{-(s^2+|N-M|s)} \label{eq:CP_MIMO2}
\end{equation}
where $\alpha_k$ is defined in~\eqref{eq:alpha} and $s,M$, and $N$ are arbitrary integers.
Define 
\[
\theta_i \defeq \sum_{\alpha_{k,i}>1}1
\]
$\theta_i$ are i.i.d. discrete random variables with the following
asymptotic distribution (cf.~Section~\ref{sec:OutageAnalysis},
Equations~\eqref{eq:OutageIntegral}-\eqref{eq:DivRes})
\begin{equation}
\prob\big( \theta_i = n_i \big) \doteq \rho^{-(n_i^2+|N-M|n_i)} \quad \text{for  } n_i=1,\dots,M \label{eq:CP_DRV}
\end{equation}
Using~\eqref{eq:CP_MIMO2}, the outage probability in~\eqref{eq:CPEig} can be evaluated as
\begin{align}
P_{out} &\dot{\leqslant}  \prob \bigg(
\sum_{i=1}^{L_d} \sum_{k=1}^{M}\frac{1}{(1+\rho \lambda_{k,i})} > ML_d2^{-\frac{R}{M}}
\bigg) \nonumber\\
&\dot{=}\prob \big( \sum_{i=1}^{L_d} \theta_i \geqslant \Omega \big) \label{eq:CP_prob}
\end{align}
where $\Omega=\lceil ML_d2^{-\frac{R}{M}}\rceil$.
Evaluating the probability in~\eqref{eq:CP_prob} in a
combinatorial manner, we get
\begin{align}
&\prob \big( \sum_{i=1}^{L_d} \theta_i \geqslant \Omega \big)\doteq\prob
\big( \sum_{i=1}^{L_d} \theta_i = \Omega
\big) \nonumber\\
&\dot{=}\sum_{n_1,n_2,\dots,n_{L_d}} \rho^{-(n_1^2+|N-M|n_1)} \dots \rho^{-(n_p^2+|N-M|n_{L_d})}\label{eq:CP_Combinator}\\
&\dot{=} \max_{n_1,n_2,\dots,n_{L_d}} \rho^{-(n_1^2+|N-M|n_1)} \dots \rho^{-(n_{L_d}^2+|N-M|n_{L_d})}\label{eq:CP_max}
\end{align}
where $n_i\in[0,M]$ for ($i=1,2,\dots,L_d$) is the value of the $i$-th
discrete random variable $\theta_i$, and~\eqref{eq:CP_max} is true
since the summation in~\eqref{eq:CP_Combinator} is dominated by the
maximum element.

Let the set $\{n_k^*, k=1,\dots,L_d\}$ be the set of indices
of the optimal solution of~\eqref{eq:CP_max}. The set $\{n_k^*\}$ is
obtained by solving the following optimization problem
\begin{equation}
\min_{n_1,n_2,\dots,n_{L_d}} \sum_{k=1}^{L_d}(n_k^2+|N-M|n_k)\nonumber
\end{equation}
\begin{align}
\text{subject to}\quad \quad   \sum_{k=1}^{L_d} n_k =\Omega \nonumber\\
            0 \leqslant n_k \leqslant M \nonumber
\end{align}
or equivalently,
\begin{align}
\min_{n_1,n_2,\dots,n_{L_d}} \quad &\sum_{k=1}^{L_d} n_k^2 \label{eq:ZP_Optim}\\
\text{subject to }  &\sum_{k=1}^{L_d} n_k =\Omega \nonumber\\
&n_k \ge 0 \nonumber
\end{align}

The problem in~\eqref{eq:ZP_Optim} is a quadratic integer-programming
(QIP) problem (see e.g.~\cite{Bradely77:book} ). Integer programming
problems are in general NP-hard. However, due to the simple structure
of the objective function in~\eqref{eq:ZP_Optim}, we can efficiently
solve it, thus obtain a closed form expression
for $\{n_k^*\}$ and hence~\eqref{eq:CP_max}. 

\begin{lemma}
\label{Lemma:CP_Optimization}
For the QIP given by~\eqref{eq:ZP_Optim}, the optimum solution is given by:
\begin{align}
&n_i^*=u \quad \text{for } 1\leqslant i \leqslant t \nonumber\\
&n_j^*=u+1 \quad \text{for } t+1\leqslant j \leqslant L_d \nonumber
\end{align}
where $u=\lfloor \frac{\Omega}{L_d} \rfloor$ and $t=L_d(u+1)-\Omega$.

\begin{proof}
See Appendix~\ref{appendix:QIP}
\end{proof}
\end{lemma}
 
Using Lemma~\ref{Lemma:CP_Optimization}, we can now evaluate the outage upper bound given by~\eqref{eq:CP_max} as 
\begin{equation}
P_{out} \hspace{4pt} \dot{\leqslant} \hspace{4pt}\rho^{-d_{cp}} \label{eq:CPDivLenTaps}
\end{equation}
where $d_{cp}=\Omega (2u+1)-uL_d(u+1)+|N-M|\Omega$ and $u=\lfloor \frac{\Omega}{L_d} \rfloor$.

\subsubsection{Outage lower bound}
The bound is obtained using the same steps to obtain the lower bound in
Section~\ref{sec:ZPLB}. It can be shown that 
\begin{align}
P_{out}&=\prob \bigg( \frac{1}{L_d} \sum_{k=1}^{ML_d} \log
((\bI+\rho\bHeq^H\bHeq)^{-1}_{kk}) > -R \bigg)
\label{eq:CPMMSEDiv}\\
%%  &\geqslant \prob \bigg( M\log \sum_{k=1}^{ML_d}
%% \frac{1}{ML_dM_{\beta}}(\bI+\rho\bHeq^H\bHeq)^{-1}_{kk} > -R
%% \bigg) \label{eq:CPSpecht}\\ &=  \prob \bigg( M\log
%% \sum_{k=1}^{ML_d} \frac{1}{ML_dM_{\beta}}(\bI+\rho{\bf \Lambda}^H{\bf \Lambda})^{-1}_{kk} > -R
%% \bigg)\nonumber\\
%% &=  \prob \bigg(
%% \sum_{i=1}^{L_d} \frac{1}{M_{\beta}}\sum_{k=1}^{M}\frac{1}{(1+\rho \lambda_{k,i})} > ML_d2^{-\frac{R}{M}}
%% \bigg) \nonumber\\
&\dot{\geqslant}\hspace{5pt}   \prob \bigg(
\sum_{i=1}^{L_d} \sum_{k=1}^{M}\frac{1}{(1+\rho \lambda_{k,i})} > ML_d2^{-\frac{R}{M}}
\bigg) \label{eq:CPEigSpechtLower}
\end{align}
%% where~\eqref{eq:CPSpecht} follows from Specht's inequality (cf.
%% Section~\ref{sec:LB}) with $M_{\beta}$ as given
%% by~\eqref{eq:SpechtRatio}, and~\eqref{eq:CPEigSpechtLower} follows
%% similarly to Section~\ref{sec:LB}. 
The bound
in~\eqref{eq:CPEigSpechtLower} is the same as the upper bound in~\eqref{eq:CPEig}, thus the bound is tight and the diversity is given
by~\eqref{eq:CPDivLenTaps}.  The PEP analysis follows in a
manner similar to Section~\ref{sec:PEPAnalysis}.

Recall that so far we have considered data block length $L_d=\nu+1$.
It can be shown that the diversity for any $L_d>\nu+1$ is upper
bounded by the computed diversity for the case $L_d=\nu+1$. This
bounding is derived from~\eqref{eq:CPDiag} via FFT arguments
similar to those used in~\cite{Tajer:WCOM10}, which we omit for
brevity. A tight diversity {\em lower} bound for data block lengths
$L_d>\nu+1$ remains an open problem, except for the SIMO system as
discussed in the next section.

%------------------------------------------------------------------------
\subsubsection{Diversity of CP Transmission in the SIMO Channel}
\label{sec:SpecialI}
%% The diversity of SIMO system under CP transmission is precisely
%% characterized in the following theorem.

 \begin{theorem}
\label{th:CPSIMO}
In a SIMO quasi-static frequency-selective channel with memory
$\nu$, $N$ receive antennas and data-block length $L_d$, the MMSE
receiver diversity is $d^{CP}_{MMSE}=N \min(\nu+1,\lfloor 2^{-R}L_d \rfloor+1)$ under
joint spatial encoding and cyclic prefix transmission.
\end{theorem} 

In order to prove Theorem~\ref{th:CPSIMO}, we first analyze the case
of $L_d=\nu +1$ and then generalize the result for $L_d>\nu+1$.
The system model is given by~\eqref{eq:Modifmodel} where the $NL_d\times
L_d$ equivalent channel matrix is given by
\begin{equation}
\bHeq=
\begin{bmatrix} \bh_0 &\bh_1 &\cdots &\bh_{\nu} &0 &\cdots &0 \\
                0  &\bh_0 &\bh_1 &\cdots &\bh_{\nu} &\cdots &0 \\
                \vdots  &\ddots &\ddots &\ddots &\ddots  &\ddots &\vdots \\
                \bh_{1} &\bh_{2}  &\cdots &\bh_{\nu}  &0  &\cdots &\bh_{0}
\end{bmatrix}.
\end{equation}
where $\bh_i$ (for $i=0,1,\dots,\nu$) is $N\times 1$ SIMO channel.
Note that the diagonal elements of ($\bHeq^H\bHeq$) are identical and
equal to $\sum_{i=0}^{\nu}\bh_i^H\bh_i$. Thus the MMSE SINR for each
output information stream is
\begin{equation}
\gamma_k=\frac{1}{(\bI+\bHeq^H\bHeq)_{kk}}-1= \frac{1}{\frac{1}{L_d}\text{tr}(\bI+\bHeq^H\bHeq)_{kk}}-1 \label{eq:CPMMSESINR}
\end{equation}

Evaluating the outage probability as in~\eqref{eq:CPMMSEOuEq}
\begin{align}
P_{out}&=\prob \bigg( \frac{1}{L_d} \sum_{k=1}^{L_d} \log
(\frac{1}{(\bI+\rho\bHeq^H\bHeq)^{-1}_{kk}}) < R \bigg)
\nonumber\\ &= \prob \bigg( \log
\frac{1}{L_d} \sum_{k=1}^{L_d}\frac{1}{(\bI+\rho\bHeq^H\bHeq)^{-1}_{kk}} < R
\bigg) \label{eq:CPJenSIMO}\\
&=  \prob \bigg(
\sum_{k=1}^{L_d}\frac{1}{(1+\rho \lambda_{k})} > L_d2^{-R}
\bigg) \label{eq:CPEigSIMO}
\end{align}
where~\eqref{eq:CPJenSIMO} follows from~\eqref{eq:CPMMSESINR}
and~\eqref{eq:CPEigSIMO} follows similarly to~\eqref{eq:CPEig}.

In a manner similar to~\eqref{eq:CPEig} we have $\lambda_k={\bf
  B}_k^H{\bf B}_k$ because now $\mathbf B$ is simply a $N\times 1$
vector. For the case $L_d=\nu+1$, the eigenvalues $\{\lambda_k\}$ are
distributed according to Gamma distribution with shape parameter $N$
and scale parameter $1$, i.e. $\lambda_k \sim {\Gamma}(N,1)$. For
$L_d>\nu+1$ the Gaussian variables in ${\mathbf B}_k$ are no longer
independent and thus analyzing this case requires the unknown
distribution $\{\lambda_k\}$. Instead, we indirectly show that the
diversity of $L_d=\nu+1$ also holds for $L_d>\nu+1$.

%% {\bf case $P=\nu+1$:}

 \begin{lemma}
\label{Lemma:CPSIMOSpecialCase}
In a SIMO quasi-static frequency-selective channel with memory
$\nu$, $N$ receive antennas and data-block length $L_d=\nu+1$, the MMSE
receiver diversity is $d^{CP}_{MMSE}={N(\lfloor L_d2^{-R} \rfloor+1)}$ under
joint spatial encoding and cyclic prefix transmission.
\end{lemma}

\begin{proof}

The outage probability can be written as
\begin{align}
P_{out}&=\prob \big(\sum_{k=1}^{L_d}\frac{1}{(1+\rho \lambda_{k})} > L_d2^{-R}
\big) \nonumber\\
&\dot{=} \hspace{6pt}\prob \big(M({\mathbf\alpha}) > L_d2^{-R}
\big) \label{eq:OutSimoCP}
\end{align}
where we use $M({\mathbf\alpha}) = \sum_{\alpha_k>1}1 \label{eq:SIMOCPOnes}
$ from~\eqref{eq:CP_MIMO}.
We thus need to evaluate $\prob(\alpha >1)$.  The probability density function of
$\lambda_k$ is
\begin{equation}
f_{\lambda_k}(x)=\frac{1}{\Gamma(N)}\, x^{N-1}\, e^{-x}\end{equation}
The distribution of $\alpha_k$ is thus given by
\begin{equation}
f_{\alpha_k}(x)=\frac{1}{\Gamma(N)} \, \rho^{-Nx} \, e^{-x} \, \ln{\frac{1}{\rho}}
\end{equation}
The cumulative distribution function of $\alpha_k$ is
\begin{align}
F_{\alpha_k}(x)&=\int_{-\infty}^{x} f_{\alpha_k}(y) \; dy\nonumber\\
&=\frac{1}{\Gamma(N)}\int_{\rho^{-x}}^{\infty} r^{N-1} e^{-r} dr\label{eq:CPSIMOGamma}\\
&=\frac{1}{\Gamma(N)} \bigg (  \int_{0}^{\infty}r^{N-1} e^{-r} dr - \int_0^{\rho^{-x}} r^{N-1} e^{-r} dr \bigg )\label{eq:CPSIMOGamma2}\\
&= e^{-\rho^{-\alpha}}\sum_{k=0}^{N-1}\frac{\rho^{-x k}}{k!} \label{eq:CPSIMOGamma3}
\end{align}
where we have made a change of variables $r=\rho^{-x}$
in~\eqref{eq:CPSIMOGamma}, and evaluate the integral according
to~\cite[P.334 and P.336]{Gradshtey:book}.
Thus we have 
\begin{align}
P(\alpha_k>1)&=1-e^{-\rho}\sum_{k=0}^{N-1}\frac{\rho^{-k}}{k!}\nonumber\\
&\doteq 1-\big ( 1-\frac{1}{N!}\rho^{-N} \big)\label{eq:CDFSIMO}\\
&\doteq \rho^{-N} \label{eq:CDFSIMO2}
\end{align}
where~\eqref{eq:CDFSIMO} follows from the Taylor expansion
for~\eqref{eq:CPSIMOGamma3}.

From the independence of $\{\lambda_k\}$, and subsequently the
independence of $\{\alpha_k\}$, we conclude that $M({\mathbf{\alpha}})$
in~\eqref{eq:OutSimoCP} is binomially distributed with parameter $\rho^{-N}$. Hence, similar to~\cite{Tajer:WCOM10}, we have
\begin{align}
&\prob \bigg(  \sum_{k=1}^{L_d} \frac{1}{1+\rho\lambda_k}> L_d2^{-R}  \bigg) \doteq  \prob(M(\alpha)>L_d2^{-R})\nonumber\\
&=\sum_{i=\lfloor L_d2^{-R} \rfloor+1}^{L_d} \prob(M(\alpha)=i)\nonumber\\
&\doteq  \sum_{i=\lfloor L_d2^{-R} \rfloor+1}^{L_d} 
\begin{pmatrix}
{L_d}\\ 
{i}
\end{pmatrix}
\rho^{-Ni} \underbrace{(1-\rho^{-N})^{n-i}}_{\doteq 1}\nonumber \\
&\doteq \rho^{-N{(\lfloor L_d2^{-R} \rfloor+1)}}\nonumber.
\end{align}
which concludes the proof for $L_d=\nu+1$
\end{proof}

%% {\bf case $P>\nu+1$:}

For $L_d>\nu+1$ we follow steps similar to~\cite{Tajer:WCOM10}. 

 \begin{lemma}
\label{Lemma:CPSIMOPadding}

Consider two SIMO systems both operating under quasi-static
frequency-selective channels with memory $\nu$.  One system has
data block length $L_{d_1}>\nu+1$ and the other $L_{d_2}\geq L_{d_1}$,
we have the following property
\begin{equation}
 \prob \bigg(
\sum_{k=1}^{L_{d_1}}\frac{1}{(1+\rho \lambda_{k})} > m
\bigg) \doteq \prob \bigg(\sum_{k=1}^{L_{d_2}}\frac{1}{(1+\rho \lambda_{k})} > m
\bigg)  \nonumber
\end{equation}
for any $m\in \mathbb{R}$.
\end{lemma}
\begin{proof}
The proof has similarities with the SISO case developed in~\cite[Lemma
  2]{Tajer:WCOM10}, but is not a trivial extension (see Appendix~\ref{appendix:CPPAdding}).
\end{proof}

Using Lemma~\ref{Lemma:CPSIMOPadding} and the results in~\cite[Theorem
  2]{Tajer:WCOM10}, Theorem~\ref{th:CPSIMO} is established.

%--------------------------- Simulation Results --------------------------------

\section{Simulation Results} 
\label{sec:SimRes}

%We consider MIMO system under flat Rayleigh fading with different
%number of transmit and receive antenna. 
Simulations generate Monte Carlo random channel realizations and
calculate outage probability by checking the appropriate linear MIMO
receiver mutual information for the quasi-static flat fading
model. Figure~\ref{Fig1} shows the case $M=N=3$.
%The target rate is $R$ is ranged from 1 to 10 bps/Hz. 
According to
Theorem~\ref{theorem:MMSE_MIMO}, $d_{out}=1$ for $R\geqslant 4.755$, $
d_{out}=4$ for $4.755 > R\geqslant 1.755$, and $ d_{out}=9$ for $ R <
1.7549$.  Figure~\ref{Fig1} shows the diversity step between $R=4.5$
and $4.8$bps/Hz. The slope of diversity 9 is difficult to measure
precisely with simulations, but it is approximately
observed. Figure~\ref{Fig2} shows the outage probability for $R=1, 4
\text{ and } 10$ with the Jensen bound, with a diversity
transition at $R=2$. Figure~\ref{Fig3} shows the case of $M=2, \text{
  and }N=3$ again with transition at $R=2$.
In Figure~\ref{Fig4}, simulations results for $N=2$ and $M=3$ are given and compared with $N=3$ and $M=2$. Theorem~\ref{theorem:MMSE_MIMO} gives the diversity for both systems. It is observed that when $N>M$ the break point of the slopes occurs before its counterparts in $M>N$ case. Lower rates were difficult to simulate precisely.

\begin{figure}
\centering
\includegraphics[width=3.5in]{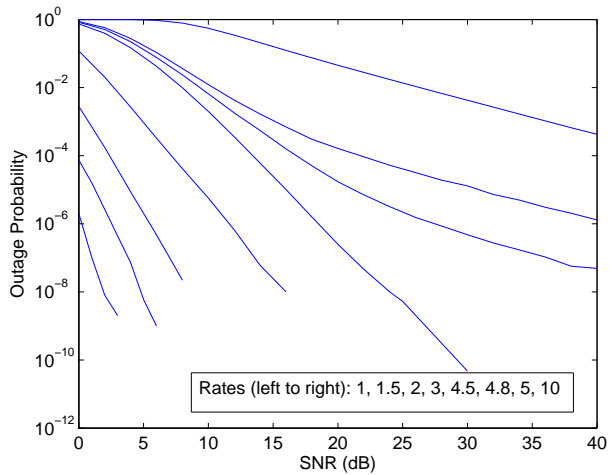}
\caption{Outage probability of MMSE Receiver, $M=N=3$ for R=1, 1.5, 2, 3, 4.5, 4.8, 5, 10 bps/Hz}
\label{Fig1}
\end{figure}

\begin{figure}
\centering
\includegraphics[width=3.5in]{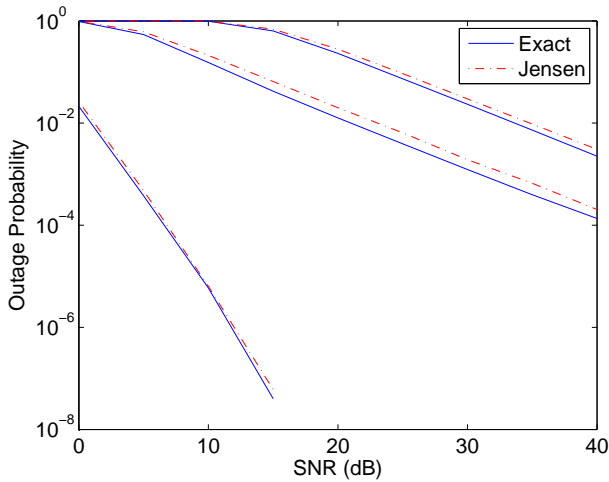}
\caption{Outage probability of MMSE Receiver, $M=N=2$ for R (left to
  right)= 1, 4, 10 bps/Hz}
\label{Fig2}
\end{figure}

\begin{figure}
\centering
\includegraphics[width=3.5in]{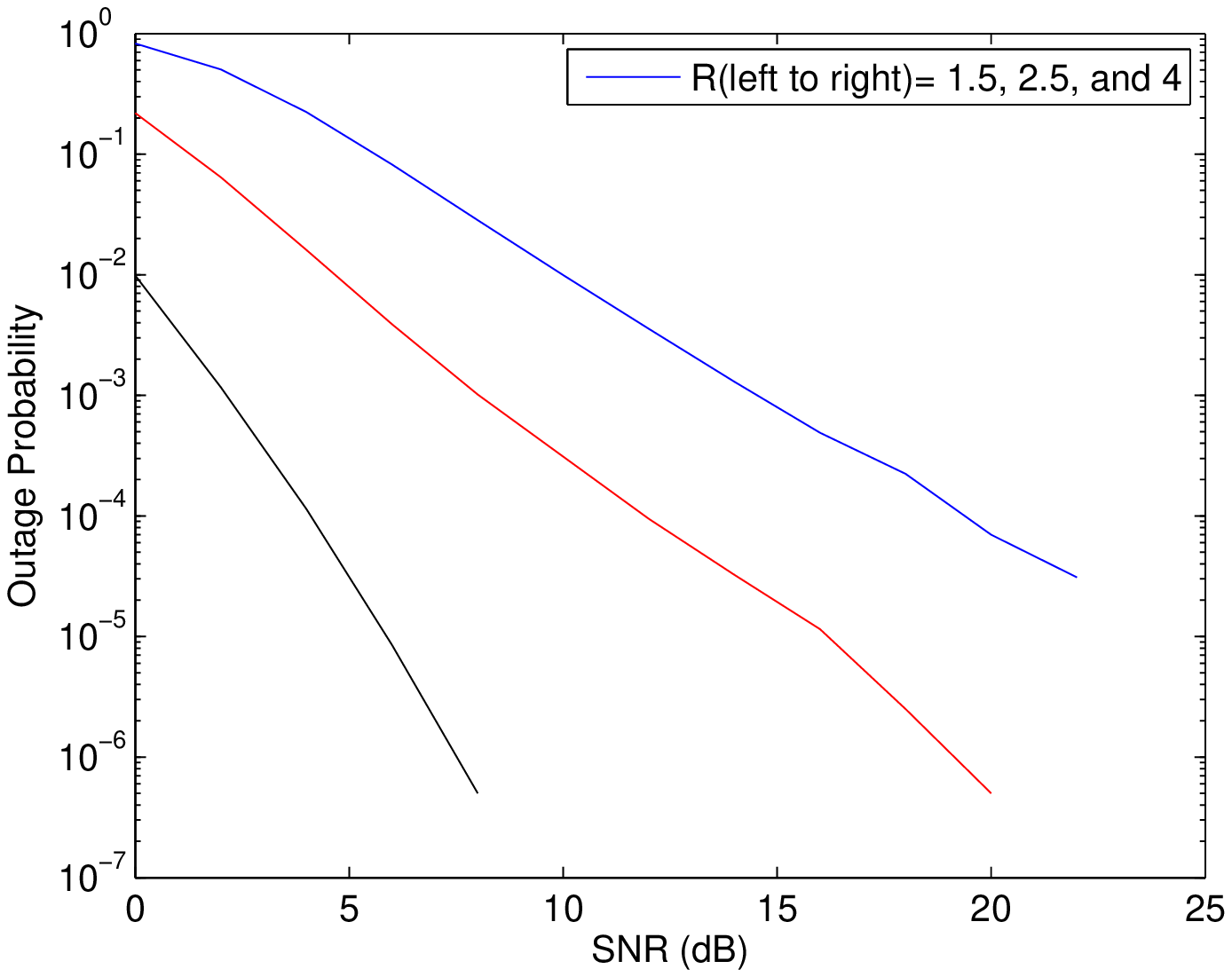}
\caption{Outage probability of MMSE Receiver, $M=2, \text{ and }N=3$
  for R (left to right)= 1.5, 2.5, 4 bps/Hz}
\label{Fig3}
\end{figure}

\begin{figure}
\centering
\includegraphics[width=3.5in]{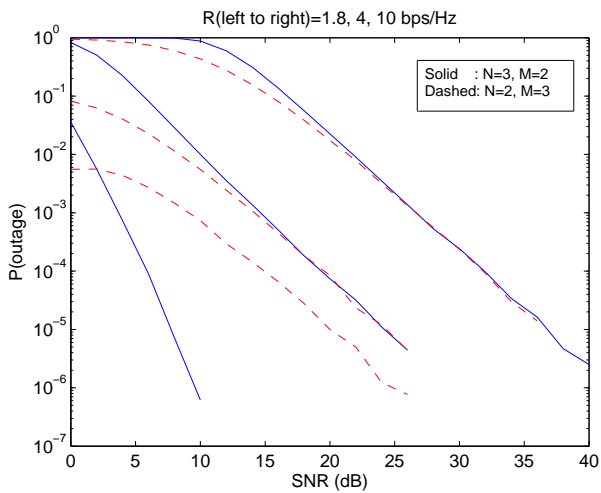}
\caption{Outage probability of MMSE Receiver for both cases $N>M$
  (solid) and $M>N$ (dashed). The spectral efficiency $R$ (left to right)= 1.8, 4, and 10 bps/Hz}
\label{Fig4}
\end{figure}

%---------------------------------------------------------------------------------
%------------------------------ Conclusion   ------------------------------------
%---------------------------------------------------------------------------------
\section{Conclusion} 
\label{sec:Conclusion}

This paper settles the long standing problem of the diversity of the
MMSE MIMO receivers under all fixed rates for any number of transmit
($M$) and receive ($N$) antennas, giving the result as $d=\lceil M
2^{-\frac{R}{M}}-(M-N)^+\rceil^{2}+ |N-M|\lceil M
2^{-\frac{R}{M}}-(M-N)^+\rceil$. The analysis confirms the earlier
approximate results~\cite{Hedayat:JSP07,Kumar:JIT09} showing that the
system diversity can be as high as $MN$ for low spectral efficiency
and as low as $N-M+1$ for high spectral efficiency. The result is
extended to the multiple access channel (MAC). We also analyze the
case of frequency-selective MIMO channel under cyclic-prefix and
zero-padding transmission, and obtain the explicit tradeoff between
rate and diversity.

%-------------------------------------------------------------------------------
%--------------------------------- Appendix:Specht  -------------------------------------
\appendix

\subsection{Proof of Lemma~\ref{lemma:Div_CP_ZP_ISI}}\label{appendix:CP_ZP_ISI_APP}

Consider a single-antenna ISI channel $\mathbf{h} =[h_0, \dots ,
  h_{\nu}]$, where $\nu$ is channel memory. The transmitter sends a
block of $L_d+\nu$ symbols (i.e. the extension $L_e=\nu$), the last
$\nu$ symbols of which are zeros to remove the inter-block
interference. %% The diversity of cyclic-prefix
%% transmission was previously analyzed in~\cite{Tajer:WCOM10} where it
%% was shown that the MMSE receiver achieves a maximum diversity of
%% $d^{CP}_{MMSE}=\min(\nu+1,\lceil 2^{-R}P\rceil)$. We now proceed to
%% analyze the ZP case.
 The system model is given by
\begin{equation}
\by=\mathbf{\bHeq}\bx+\bn \label{eq:ISIModel}
\end{equation}
where $\bx$ is the transmitted length-$(L_d+\nu)$ vector. We consider
the case where the padding length is equal to the memory of the
channel. The results are also valid for $L_e>\nu$ as a direct result
of~\cite[Theorem 2]{Tajer:WCOM10}.

  The outage probability
of MMSE receiver under ZP transmission is given by~\cite{Tajer:WCOM10}
\begin{align}
P_{out}&=\prob \big( \frac{1}{L_d} \sum_{k=1}^{L_d} \log
(\frac{1}{(\bI+\rho\bHeq^H\bHeq)^{-1}_{kk}}) < R \big)
\nonumber\\
&\leqslant\prob \big( \frac{1}{L_d} \sum_{k=1}^{L_d} \log
(1+\frac{\rho}{(\bHeq^H\bHeq)^{-1}_{kk}}) < R \big)
\label{eq:ZFB}\\ 
&\leqslant \prob \big( \log \frac{1}{L_d}\sum_{k=1}^{L_d}
\frac{1}{\rho}(\bHeq^H\bHeq)^{-1}_{kk} > -R
\big) \label{eq:Jen}\\
&= \prob \big( \frac{L_d \hspace{2pt} 2^{-R}}{\text{tr}
(\bHeq^H\bHeq)^{-1}} < \rho^{-1}
\big) \label{eq:ZFBound}
\end{align}
where~\eqref{eq:ZFB} represents the outage probability of zero-forcing
equalizer which upper bounds that of the MMSE. The bound
in~\eqref{eq:Jen} follows from Jensen's inequality.

We want to show that $ {\text{ tr}} (\bHeq^H\bHeq)^{-1}$
in~\eqref{eq:ZFBound} is proportional to $||{\bf h}||^{-2}$. Thus it is
straightforward to obtain full-diversity at any $R$
since~\cite{Zheng:JIT03}
\begin{equation}
 \prob \big ( c \hspace{2pt}||{\bf h}||^2  < \rho^{-\alpha} \big)\dot{=} \rho^{-L\alpha} \label{eq:asymtot}
\end{equation}
where $c$ is a constant that is independent of ${\bf h}$.

To show that this is indeed the case, we use the result of
Tepedelenlioglu~\cite{Tepedel:ICASSP03,Tepedel:JIT04} which provides a
family of linear zero-forcing equalizers that is capable of achieving
full multipath diversity in zero-padded systems under certain
constraints. We paraphrase the result  for convenience.

\begin{lemma}[\cite{Tepedel:ICASSP03,Tepedel:JIT04}]\label{lemma:Tepedel}
Under zero-padded transmission, there exists a family of left-inverses
of $\bHeq$, denoted by ${\bf G}$, such that $||{\bf G}||^{-2}
\geqslant C ||\bh||^2$ for some constant $C$ independent of the
channel vector $\bh$. Moreover, we have $||\bW_{ZF}|| \leqslant ||{\bf
  G}||$, for any ${\bf G}$ satisfying ${\bf G} \bHeq=\bI$, and
$\bW_{ZF}$ is given by
\begin{equation}
 \bW_{ZF} = (\bHeq^{H}\bHeq)^{-1}\bHeq^{H}. \label{eq:ZFEqMod}
\end{equation}
\end{lemma}

Applying the ZF equalizer $\bW_{ZF}$ on the channel output given by~\eqref{eq:ISIModel} we get the equalized signal $\tilde{\by}=\bx+\bz$, where $\bz=\bW_{ZF}\bn$. 
The filtered noise power $P_z$ can be evaluated as
\begin{align}
 P_z &=\expect \hspace{2pt} \text{tr} [\bz\bz^H] \nonumber \\
     &=  \text{tr} \big[\expect ( (\bHeq^H\bHeq)^{-1}\bHeq^H \bn\bn^H \bHeq(\bHeq^H\bHeq)^{-1})  \big] \nonumber\\
&=\text{tr} [ (\bHeq^H\bHeq)^{-1}] \label{eq:noiseVar1}
\end{align}
where we assume the noise is uncorrelated and has variance equal to
one.

Using the properties of the Frobenius norm, $P_z$ can be bounded as
\begin{align}
 P_z &=\expect (||\bW_{zf}\bn||^2) \nonumber \\
     &\leqslant \expect(||\bW_{zf}||^2 ||\bn||^2 ) = L_d  ||\bW_{zf}||^2.  \label{eq:noiseVar2}
\end{align}

Using~\eqref{eq:noiseVar1}, \eqref{eq:noiseVar2} and Lemma~\ref{lemma:Tepedel},  the trace in~\eqref{eq:ZFBound} can be bounded by
\begin{equation}
  \text{tr} [ (\bHeq^H\bHeq)^{-1}] 
\leqslant L_d  ||\bW_{zf}||^2 \leqslant \frac{L_d} { C \hspace{2pt} ||{\bf h}||^2}. \label{eq:NoiseBound}
\end{equation}

Thus from~\eqref{eq:ZFBound} we have
\begin{align}
P_{out} &\dot{\leqslant} \prob \big( C_2||{\bf h}||^2 < \rho^{-1}
\big)\nonumber\\ 
&\dot{=} \quad \rho^{-(\nu+1)}. \label{eq:ZFZPB}
\end{align}
where $C_2=C\;2^{-R}$ is a constant independent of $\bh$ and $\rho$.

Note that the constraints and construction methods
in~\cite{Tepedel:ICASSP03,Tepedel:JIT04} for the zero-forcing
equalizers to achieve full multipath diversity in ZP
systems do not apply in CP systems. That is, Lemma~\ref{lemma:Tepedel}
is not true for CP transmission. This is because the equivalent
channel in CP systems does not have the same properties that were used
in~\cite{Tepedel:ICASSP03,Tepedel:JIT04}.

%%%%%%%%%%%%%%%%%%%%%%%%%%%%%%%%%%%%%%%%%%%%%%%%%%%%%%%%%%%%%%%%%%%%%%%%%%%%%%%%

\subsection{Proof of Lemma~\ref{Lemma:CP_Optimization}:(QIP Problem)}\label{appendix:QIP}

Consider the following Quadratic Integer Programming (QIP) problem
\begin{align}
\min_{n_1,n_2,\dots,n_{\ell}} \quad& \sum_{k=1}^{\ell} n_k^2 \label{eq:ZP_Optim2}\\
\text{subject to }  &\sum_{k=1}^{\ell} n_k =\Omega \nonumber\\
&n_k\ge 0. \nonumber
\end{align}
where  $\Omega$ and $\ell$ are integers. 

 Consider a candidate solution vector $[n_1, \ldots,n_k, \ldots,
   n_{\ell}]$. We partition the variables in this vector according to
 their values into $\Omega+1$ sets $\mathcal{N}_j = \{ n_k : n_k =
 j\}$ for $0\le j \le \Omega$; clearly some of these sets may be
 empty. Denote the membership of each set $S_j = |\mathcal{N}_j|$.
 Furthermore, let $\Omega = m \ell + K$ where $m$ is the divisor and
 $K$ is the remainder of the division of $\Omega$ by $\ell$. From the constraint
in~\eqref{eq:ZP_Optim2} we have
\begin{equation}
\sum_{k=1}^{\ell}n_k= \sum_{j=0}^\Omega j S_j = m\ell+\sum_{j=0}^{\Omega} (j-m)S_j=m\ell+K\label{eq:AppenQIP}.
\end{equation} 
Evaluating the objective function:
\begin{align}
\sum_{k=1}^{\ell}n_k^2&=\sum_{j=0}^\Omega(m+j-m)^2S_j\nonumber\\
&=\ell m^2+2m\sum_{j=0}^\Omega (j-m)S_j+\sum_{j=0}^\Omega (j-m)^2S_j\nonumber\\
&=\ell m^2+2mK+\sum_{j=0}^\Omega (j-m)^2S_j\label{eq:QIPEq}\\
&\geqslant \ell m^2+2mK+\sum_{j=0}^\Omega (j-m)S_j\label{eq:QIPBound}\\
&=\ell m^2+2mK+K\label{eq:QIPBound2}
\end{align}
where \eqref{eq:QIPEq} and \eqref{eq:QIPBound2} use
$\sum_{j=0}^{\Omega} (j-m)S_j=K$, which follows from \eqref{eq:AppenQIP}.

We now propose that one may achieve optimality when all variables take
values either $m$ or $m+1$. In that case,
\begin{align}
\sum_kn_k&=m S_m +(m+1)(\ell-S_m)=m \ell+(\ell-S_m)\nonumber\\
\sum_kn_k^2&=m^2 S_m +(m+1)^2(\ell-S_m)=\ell m^2+2mK+K\nonumber.
\end{align}
where we substituted the value of $\ell-S_m$ from the first equation
into the second equation above. This shows that the variables taking
values $m$ or $m+1$ achieves the lower bound in~\eqref{eq:QIPBound2}. At optimality $S_m = (m+1)\ell -\Omega$.

%%%%%%%%%%%%%%%%%%%%%%%%%%%%%%%%%%%%%%%%%%%%%%%%%%%%%%%%%%%%%%%%%%%%%%%%%%%%%%%%5

\subsection{Proof of Lemma~\ref{Lemma:CPSIMOPadding}}\label{appendix:CPPAdding}

We begin by showing that for any integer multiplier of $L_{d_1}=\nu+1$
denoted by $L_{d_2}=TL_{d_1}$ ($T\in \mathbb{N}$) and any real-valued
$m\in(0,L_{d_1})$, we have
\begin{equation}
 \prob \bigg(
\sum_{q=1}^{L_{d_1}}\frac{1}{(1+\rho \lambda_{q})} > m
\bigg) \doteq \prob \bigg(\sum_{q=1}^{L_{d_2}}\frac{1}{(1+\rho \lambda_{q})} > m
\bigg)  \label{eq:EqProof}
\end{equation}

Note that for SIMO-CP system, $\lambda_q={\mathbf b}_q^H{\mathbf b}_q$, where
${\mathbf b}_q$ is the $N\times 1$ vector given by
\begin{equation}
{\mathbf b}^{(i)}_q=\sum_{n=0}^{\nu}\bh_n \, e^{-j\frac{2\pi (q-1)}{L_{d_i}}}\quad {\text{for }} q=1,\dots,L_{d_i} \label{eq:LambdaEq2}
\end{equation}
where $\bh_n$ is the channel gain as a function of the tap
index $n$, and the superscript $i=1,2$ is used to distinguish the variables
in two systems with data block lengths $L_{d_1}$ and $L_{d_2}$.

Recall that we can take a $L_{d_1}$-point signal and apply a
$L_{d_2}$-point DFT on it (after zero-padding), which will result in a
resampling in the Fourier domain at $L_{d_2}$
points. Following~\cite{Tajer:WCOM10} we can write the explicit
relationship between entries of ${\mathbf b}^{(1)}$ and ${\mathbf b}^{(2)}$ as
\begin{equation}
{b}^{(1)}_{q,l}=\sum_{i=1}^{L_{d_1}} {b}^{(2)}_{i,l} \psi_i \quad  q=1,2\dots,L_{d_2}\text{ and } l=1,2,\dots,N\label{eq:DFTRelation}.
\end{equation}
where
\begin{equation}
 \psi_i=\frac{1}{L_{d_1}}\frac{1-e^{-j\frac{(q-1)2\pi L_{d_1}}{L_{d_2}}}}{1-e^{-j\big( \frac{2\pi(q-1)}{L_{d_2}} - \frac{2\pi(i-1)}{L_{d_1}}  \big)}} \nonumber.
\end{equation}

Define
$\alpha^{(i)}_{q,l}=\-\frac{\log{ |b^{(i)}_{q,l}|^2}}{\log{\rho}}$.
 Note that
$b^{(1)}_{T(q-1),l}=b^{(2)}_{q,l}$ and
$\alpha^{(2)}_{T(q-1),l}=\alpha^{(1)}_{q,l}$ for $q=1,2\dots,L_{d_1}$
since $L_{d_2}=TL_{d_1}$ .
From~\eqref{eq:DFTRelation}, we have
\begin{equation}
|{b}^{(1)}_{q,l}|^2=\sum_{i=1}^{L_{d_1}}|\psi_i|^2 |{b}^{(2)}_{i,l}|^2 + \underbrace{
\sum_{i=1}^{L_{d_1}} \sum_{s=1}^{L_{d_1}} \psi_i \psi_s {b}^{(2)}_{i,l} {b}^{*(2)}_{s,l}
}_{\triangleq \,\,\eta}.
\label{eq:DFTComp}
\end{equation}

We now analyze each part of the sum in~\eqref{eq:DFTComp}. For the set
of indices $\mathcal{A}\triangleq\{i:i=T(k-1)+1,k=1,\dots,L_{d_1}\}$,
the coefficients $\{ \psi_i \}$ are non-zero constants, then
$|\psi_i|^2 |b^{(2)}_{i,l}|^2 \doteq |b^{(2)}_{i,l}|^2 \quad\forall
l$.  Noting that $\eta$ must be real-valued, and defining
$\alpha_{\eta}\triangleq -\frac{\log{|\eta|}}{\log \rho}$,
Eq.~\eqref{eq:DFTComp} can be written as
\begin{align}
\rho^{-\alpha^{(2)}_{q,l}} &\dot{=} \sum_{i=1}^{L_{d_1}} \rho^{-\alpha^{(1)}_{i,l}}+\frac{\eta}{|{\eta}|}\rho^{-\alpha_{\eta}} \nonumber\\
 &\dot{=} \rho^{-\min_{i} \alpha^{(1)}_{i,l}}+\frac{\eta}{|{\eta}|}\rho^{-\alpha_{\eta}} \label{eq:minalpha}.
\end{align}

Note that if $\eta < 0$ the second term in~\eqref{eq:minalpha} should
be smaller than the first term since otherwise the right-hand side
of~\eqref{eq:minalpha} will be negative while the left-hand side is
positive. Thus for $\eta <0$ we have $\alpha_\eta \geqslant \min_i
\alpha^{(1)}_{i.l}$. Also, for $a\geqslant 0$ we have $\rho^{-\min_{i}
  \alpha^{(1)}_{i,l}}+\frac{\eta}{|{\eta}|}\rho^{-\alpha_{\eta}}
\dot{\geqslant} \rho^{-\min_i \alpha^{(1)}_{i,l}}$. Thus we always have  $\rho^{-\min_{i}
  \alpha^{(1)}_{i,l}}+\frac{\eta}{|{\eta}|}\rho^{-\alpha_{\eta}}
\dot{\geqslant} \rho^{-\min_i \alpha^{(1)}_{i,l}}$, leading to the following lemma.
\begin{lemma} \label{lemma:Tajer}
For $\alpha^{(1)}_{q,l}$ and $\alpha^{(2)}_{q,l}$ defined above we have:
$\rho^{-\alpha^{(2)}_{q,l}}\dot{\geqslant} \rho^{-\min_i{\alpha^{(1)}_{i,l}}} \Rightarrow \alpha^{(2)}_{q,l} \leqslant {\min_i{\alpha^{(1)}_{i,l}}} $ for $q\in \mathcal{A}$.
\end{lemma}

We now partition the DFT points into two sets
$\mathcal{A}=\{T(i-1)+1,i=1,\dots,L_{d_1}\}$ and
$\mathcal{B}=\{1,\dots,L_{d_2} \} \backslash \{
T(i-1)+1,i=1,\dots,L_{d_1} \}$ We now define the event:
\[
\mathcal{D} \defeq \{ \min_i\alpha^{(1)}_{i,1} <1 \; , \; \min_i\alpha^{(1)}_{i,2} <1\; , \; \dots \; ,\;
\min_i\alpha^{(1)}_{i,N} <1 \}
\] 
and proceed to evaluate the probability
\begin{align}
 &\prob \bigg( \sum_{q=1}^{L_{d_2}}\frac{1}{(1+\rho \lambda_{q})} > m
 \bigg) = \prob \bigg( \sum_{q=1}^{L_{d_2}}  \frac{1}{1+\rho \sum_{l=1}^{N} |b^{(1)}_{q,l}|^2} > m
 \bigg)\label{eq:Appen0}\\
&=  \prob \bigg( \sum_{q\in{\mathcal{A}}}  \frac{1}{1+\rho \sum_{l=1}^{N} |b^{(1)}_{q,l}|^2} + \sum_{q\in\mathcal{B}}  \frac{1}{1+\rho \sum_{l=1}^{N} |b^{(1)}_{q,l}|^2} > m
 \bigg) \nonumber\\
&\doteq  \prob \bigg( S_1 + S_2 > m
 \bigg) \label{eq:VariStep}
\end{align}
where~\eqref{eq:Appen0} follows since $\lambda_q={\mathbf
  b}_q^H{\mathbf b}_q$ and $S_1$ and $S_2$ are given by
\begin{align}
S_1 &\triangleq\sum_{q=1}^{L_{d_1}}  \frac{1}{1+ \sum_{l=1}^{N} \rho^{1-\alpha^{(1)}_{q,l}}}\nonumber\\
S_2 &\triangleq  \sum_{q\in\mathcal{B}}  \frac{1}{1+ \sum_{l=1}^{N} \rho^{1-\alpha^{(2)}_{q,l}}}\nonumber
\end{align}

We now evaluate~\eqref{eq:VariStep}
\begin{align}
 &\prob \bigg( S_1 + S_2 > m
 \bigg)=  \prob \bigg(  S_1 + S_2  > m
 \; \bigg |\;  \mathcal{D} \bigg) \times \prob(\mathcal{D})\; + \nonumber\\
 &\hspace{100pt} \prob \bigg(  S_1 + S_2  > m
 \; \bigg | \; \bar{\mathcal{D}} \bigg)  \times \prob(\bar{\mathcal{D}}) \label{eq:Appen3} 
\end{align}

Note that subject to the event $\mathcal{D}$, we have
\[
S_2=\sum_{q\in\mathcal{B}}  \frac{1}{1+ \sum_{l=1}^{N} \rho^{1-\alpha^{(2)}_{q,l}}} \doteq 0
\]
Therefore this term can be asymptotically ignored. Also subject to $\bar{\mathcal{D}}$, we have
\[
S_1=\sum_{q=1}^{L_{d_1}}  \frac{1}{1+ \sum_{l=1}^{N} \rho^{1-\alpha^{(1)}_{q,l}} }\doteq L_{d_1}
\]
and since with probability one, $L_{d_1}\ge m$, the other (non-negative) term can be asymptotically ignored. Thus, both the terms involving the set $\mathcal{B}$ can be altogether ignored and we have:
\begin{align}
 \prob \bigg( \sum_{q=1}^{L_{d_2}}\frac{1}{(1+\rho \lambda_{q})} > m
  \bigg) &\doteq\prob \bigg( S_1 > m | \mathcal{D} \bigg) 
  \prob(\mathcal{D}) + \nonumber\\ &\qquad \prob \bigg( S_1 > m |
  \bar{\mathcal{D}} \bigg )  \prob(\bar{\mathcal{D}})\nonumber\\
&\doteq\prob \bigg( \sum_{q=1}^{L_{d_1}}\frac{1}{(1+\rho \lambda_{q})} > m
  \bigg)\nonumber
\end{align}

We have thus established~ \eqref{eq:EqProof} when $L_{d_1}|L_{d_2}$. We must now show that the same result holds for any $T'$ when $L_{d_1} \nmid T'$. To do so, let $L_{d_2}=T{'}L_{d_1}$, then we have
\begin{equation}
 \prob \bigg(
\sum_{q=1}^{L_{d_2}}\frac{1}{(1+\rho \lambda_{q})} > m
\bigg) \doteq \prob \bigg(\sum_{q=1}^{T{'}}\frac{1}{(1+\rho \lambda_{q})} > m
\bigg).  \label{eq:EqProof2}
\end{equation}
 
Using~\eqref{eq:EqProof} when $L_{d_1}|L_{d_2}$ and~\eqref{eq:EqProof2} when $T{'}|L_{d_2}$ together establishes~\eqref{eq:EqProof} for any two positive integers.

%------------------------- Bibliography ------------------------------------

\bibliographystyle{IEEEtran}
\bibliography{IEEEabrv,AHeshamProp}

\begin{thebibliography}{10}
\providecommand{\url}[1]{#1}
\csname url@rmstyle\endcsname
\providecommand{\newblock}{\relax}
\providecommand{\bibinfo}[2]{#2}
\providecommand\BIBentrySTDinterwordspacing{\spaceskip=0pt\relax}
\providecommand\BIBentryALTinterwordstretchfactor{4}
\providecommand\BIBentryALTinterwordspacing{\spaceskip=\fontdimen2\font plus
\BIBentryALTinterwordstretchfactor\fontdimen3\font minus
  \fontdimen4\font\relax}
\providecommand\BIBforeignlanguage[2]{{%
\expandafter\ifx\csname l@#1\endcsname\relax
\typeout{** WARNING: IEEEtran.bst: No hyphenation pattern has been}%
\typeout{** loaded for the language `#1'. Using the pattern for}%
\typeout{** the default language instead.}%
\else
\language=\csname l@#1\endcsname
\fi
#2}}

\bibitem{LiCioffi:JIT06}
P.~Li, D.~Paul, R.~Narasimhan, and J.~Cioffi, ``On the distribution of {SINR}
  for the {MMSE} {MIMO} receiver and performance analysis,'' \emph{{IEEE}
  Trans. Inform. Theory}, vol.~52, no.~1, pp. 271--286, Jan. 2006.

\bibitem{Moustakas:VTC09}
A.~L. Moustakas, K.~R. Kumar, and G.~Caire, ``Performance of {MMSE} {MIMO}
  receivers: A large n analysis for correlated channels,'' in \emph{Proc. IEEE
  Vehicular Technology Conference (VTC)}, Apr. 2009, pp. 1--5.

\bibitem{Boche:Eura07}
E.~A. Jorswieck and H.~Boche, ``Information theory outage probability in
  multiple antenna system,'' \emph{European Trans. on Telecommun.}, vol.~18,
  no.~3, pp. 217--233, June 2006.

\bibitem{Jiang:JIT07}
Y.~Jiang, M.~Varanasi, and J.~Li, ``Performance analysis of {ZF} and {MMSE}
  equalizers for {MIMO} systems: An {In-Depth} study of the high {SNR}
  regime,'' \emph{{IEEE} Trans. Inform. Theory}, vol.~57, no.~4, pp.
  2008--2026, Apr. 2011.

\bibitem{Onggosanusi:ICAS02}
E.~N. Onggosanusi, A.~G. Dabak, T.~Schmidl, and T.~Muharemovic, ``Capacity
  analysis of frequency-selective {MIMO} channels with sub-optimal detectors,''
  in \emph{Proc. IEEE ICASSP}, vol.~3, May 2002, pp. 2369--2372.

\bibitem{Hedayat:JSP07}
A.~Hedayat and A.~Nosratinia, ``Outage and diversity of linear receivers in
  flat-fading {MIMO} channels,'' \emph{{IEEE} Trans. Signal Processing},
  vol.~55, no.~12, pp. 5868--5873, Dec. 2007.

\bibitem{Kumar:JIT09}
K.~R. Kumar, G.~Caire, and A.~L. Moustakas, ``Asymptotic performance of linear
  receivers in {MIMO} fading channels,'' \emph{{IEEE} Trans. Inform. Theory},
  vol.~55, no.~10, pp. 4398--4418, Oct. 2009.

\bibitem{Gao:JIT98}
H.~Gao, P.~J. Smith, and M.~V. Clark, ``Theoretical reliability of {MMSE}
  linear diversity combining in rayleigh-fading additive interference
  channels,'' \emph{{IEEE} Trans. Commun.}, vol.~46, no.~5, pp. 666--672, May
  1998.

\bibitem{Gore:ASI02}
D.~Gore, A.~Gorokhov, and A.~Paulraj, ``Joint {MMSE} versus v-{BLAST} and
  antenna selection,'' in \emph{Proc. Asilomar Conference on Signals, Systems
  and Computers}, vol.~1, Nov. 2002, pp. 505--509.

\bibitem{Tajer:WCOM10}
A.~Tajer and A.~Nosratinia, ``Diversity order in {ISI} channels with
  single-carrier frequency-domain equalizer,'' \emph{{IEEE} Trans. Wireless
  Commun.}, vol.~9, no.~3, pp. 1022 --1032, Mar. 2010.

\bibitem{Verdu98:book}
S.~Verd\'{u}, \emph{Multiuser Detection}.\hskip 1em plus 0.5em minus
  0.4em\relax Cambridge University Press, 1998.

\bibitem{Telatar:EUR99}
I.~E. Telatar, ``Capacity of multi-antenna gaussian channels,'' \emph{European
  Trans. on Telecommun.}, vol.~10, pp. 585--595, Nov./Dec. 1999.

\bibitem{James:AMS64}
A.~James, ``Distributions of matrix variates and latent roots derived from
  normal samples,'' \emph{Annals of Mathematical Statistics}, vol.~35, pp.
  475--501, 1964.

\bibitem{Tse05:book}
D.~Tse and P.~Viswanath, \emph{Fundamentals of Wireless Communication}.\hskip
  1em plus 0.5em minus 0.4em\relax Cambridge University Press, 2005.

\bibitem{Zheng:JIT03}
L.~Zheng and D.~N.~C. Tse, ``Diversity and multiplexing: a fundamental tradeoff
  in multiple-antenna channels,'' \emph{{IEEE} Trans. Inform. Theory}, vol.~49,
  no.~5, pp. 1073--1096, May 2003.

\bibitem{Cover06:book}
T.~M. Cover and J.~A. Thomas, \emph{Elements of Information Theory}.\hskip 1em
  plus 0.5em minus 0.4em\relax John Wiley and Sons, 1991.

\bibitem{AlDhahir:CommLet01}
N.~Al-Dhahir, ``Single-carrier frequency-domain equalization for space-time
  block-coded transmissions over frequency-selective fading channels,''
  \emph{{IEEE} Commun. Lett.}, vol.~5, no.~7, pp. 304 --306, July 2001.

\bibitem{Sari:CommMag95}
H.~Sari, G.~Karam, and I.~Jeanclaude, ``Transmission techniques for digital
  terrestrial tv broadcasting,'' \emph{{IEEE} Commun. Mag.}, vol.~33, no.~2,
  pp. 100 --109, Feb. 1995.

\bibitem{Muquet:JC02}
B.~Muquet, Z.~Wang, G.~Giannakis, M.~de~Courville, and P.~Duhamel, ``Cyclic
  prefixing or zero padding for wireless multicarrier transmissions?''
  \emph{{IEEE} Trans. Commun.}, vol.~50, no.~12, pp. 2136 -- 2148, Dec. 2002.

\bibitem{Zhendao:WCNC02}
Z.~Wang, X.~Ma, and G.~Giannakis, ``Optimality of single-carrier zero-padded
  block transmissions,'' in \emph{Proc. IEEE Wireless Communcations and
  Networking Conference (WCNC)}, vol.~2, Mar. 2002, pp. 660 -- 664.

\bibitem{Scaglione:JSP02}
A.~Scaglione, P.~Stoica, S.~Barbarossa, G.~Giannakis, and H.~Sampath, ``Optimal
  designs for space-time linear precoders and decoders,'' \emph{{IEEE} Trans.
  Signal Processing}, vol.~50, no.~5, pp. 1051 --1064, May 2002.

\bibitem{Tepedel:ICASSP03}
C.~Tepedelenlioglu, ``Low complexity linear equalizers with maximum multipath
  diversity for zero-padded transmissions,'' in \emph{Proc. IEEE ICASSP},
  vol.~4, apr 2003, pp. 636--639.

\bibitem{Grokop:JIT09}
L.~Grokop and D.~Tse, ``Diversity-multiplexing tradeoff in {ISI} channels,''
  \emph{{IEEE} Trans. Inform. Theory}, vol.~55, no.~1, pp. 109 --135, Jan.
  2009.

\bibitem{Gradshtey:book}
I.S.Gradshteyn and I.M.Ryzbik, \emph{Tables of Integrals, Series, and
  Products}, 6th~ed.\hskip 1em plus 0.5em minus 0.4em\relax Academic Press.,
  2000.

\bibitem{Stamoulis:TSP02}
A.~Stamoulis, S.~Diggavi, and N.~Al-Dhahir, ``Intercarrier interference in
  {MIMO OFDM},'' \emph{{IEEE} Trans. Signal Processing}, vol.~50, no.~10, pp.
  2451 -- 2464, Oct. 2002.

\bibitem{Kaveh:Spring10}
\BIBentryALTinterwordspacing
A.~Kaveh and H.~Rahami, ``Block circulant matrices and applications in free
  vibration analysis of cyclically repetitive structures,'' \emph{Acta
  Mechanica}, pp. 1--12, 2010. [Online]. Available:
  \url{http://dx.doi.org/10.1007/s00707-010-0382-x}
\BIBentrySTDinterwordspacing

\bibitem{Bradely77:book}
\BIBentryALTinterwordspacing
S.~Bradely, A.~Hax, and T.~Magnanti, \emph{Applied Mathematical
  Programming}.\hskip 1em plus 0.5em minus 0.4em\relax Addison-Wesley, 1977,
  {Chapter 6}. [Online]. Available: \url{web.mit.edu/15.053/www/}
\BIBentrySTDinterwordspacing

\bibitem{Tepedel:JIT04}
C.~Tepedelenlioglu, ``Maximum multipath diversity with linear equalization in
  precoded {OFDM} systems,'' \emph{{IEEE} Trans. Inform. Theory}, vol.~50, pp.
  232--235, Jan. 2004.

\end{thebibliography}

%% \begin{biographynophoto}{Ahmed Hesham Mehana} (S'10) received his B.S. and M.S. degrees from Cairo University in 2004 and 2007, respectively, both in Electrical Engineering. He is currently pursuing the Ph.D degree in electrical engineering at the University of Texas at Dallas. He was an intern at Research in Motion Co. Ltd during 2010 and research assistant in Texas A\&M at Qatar in 2008. His current interests include MIMO precoding, linear receivers, and interference management.

%% \end{biographynophoto}

%% \begin{biographynophoto}{Aria Nosratinia}
%% (S'87-M'97-SM'04-F'10) is Jonsson Distinguished Professor of
%%   Engineering at the University of Texas at Dallas. He received his
%%   Ph.D. in Electrical and Computer Engineering from the University of
%%   Illinois at Urbana-Champaign in 1996. He has held visiting
%%   appointments at Princeton University, Rice University, and UCLA.
%%   His interests lie in the broad area of information theory and signal
%%   processing, with applications in wireless communications. He was the
%%   secretary for the IEEE Information Theory Society in 2010-2011 and
%%   was the treasurer for ISIT 2010 in Austin, Texas. He has served as editor
%%   for the IEEE Transactions on Information Theory, IEEE Transactions
%%   on Wireless Communications, IEEE Signal Processing Letters, IEEE
%%   Transactions on Image Processing, and IEEE Wireless Communications
%%   (Magazine). He has been the recipient of the National Science
%%   Foundation career award, and is a fellow of IEEE.
%% \end{biographynophoto}

 \end{document}